\documentclass[a4paper,UKenglish,cleveref, autoref, thm-restate]{lipics-v2021}

\usepackage[disable]{todonotes}
\usepackage{stmaryrd}
\usepackage{graphicx}
\usepackage{tikz}
\usepackage{cite}
\usepackage{proof}
\usepackage{sepnum}
\usepackage{wrapfig}
\usepackage{subcaption}
\usepackage{mathtools}
\usepackage{booktabs}

\usetikzlibrary{calc}
\usetikzlibrary {arrows.meta}

\newcommand{\dr}{\mathbf{r}}
\newcommand{\dl}{\mathbf{l}}
\newcommand{\mpg}[1]{\mathcal{#1}}
\newcommand{\bfn}[1]{\overline{#1}}
\newcommand{\nfr}[2][]{#2^{#1}_{\dr}}
\newcommand{\nfl}[2][]{#2^{#1}_{\dl}}
\newcommand{\nset}[1]{[#1]}
\newcommand{\eve}{\exists}
\newcommand{\Real}{\mathbb{R}}
\newcommand{\adam}{\forall}
\newcommand{\Nat}{\mathbb{N}}
\newcommand{\Natm}[1]{\nset{0, #1}}
\newcommand{\play}{\pi}
\newcommand{\stra}[2]{\mathrm{Str}_{#1}(#2)}
\newcommand{\indplay}[3]{\play^{#1, #2}_{#3}}
\newcommand{\indplaya}[2]{\play^{#1}_{#2}}
\newcommand{\denoplay}[2]{\llbracket #1 \rrbracket_{#2}}
\newcommand{\denoentry}[2]{\llbracket #1 \rrbracket_{#2}}
\newcommand{\denosdigest}[1]{\llparenthesis#1 \rrparenthesis}
\newcommand{\defeq}{\coloneqq}
\newcommand{\const}[1]{\star_{#1}}
\newcommand{\seqcomp}{\mathbin{;}}
\newcommand{\trsyb}{\mathrm{tr}}

\newcommand{\trace}[4]{\trsyb_{#1;#2, #3}(#4)}
\newcommand{\tdip}[5]{\mathrm{tdp}_{#1;#2,#3}(\mpg{#4},#5)}
\newcommand{\lplus}[2]{#2^{\downarrow #1}}

\newcommand{\compMPG}{\mathrm{CompMPG}}

\newcommand{\sybpg}[3]{\mathrm{PG}(#1, #2, #3)}

\newcommand{\oMPG}{\mathbf{oMPG}}
\newcommand{\roMPG}{\mathbf{roMPG}}
\newcommand{\roPlay}{\mathbf{roPG}}
\newcommand{\Id}[1]{\mathcal{I}_{#1}}
\newcommand{\Swap}[2]{\mathcal{S}_{#1, #2}}
\newcommand{\Int}{\mathrm{Int}}

\newcommand{\fpmonad}{T}
\newcommand{\Sets}{\mathbf{Set}}
\newcommand{\Ord}{\mathbf{Ord}}
\newcommand{\fpsemCat}{\mathbb{S}^{\mathbb{P}}_{\dr}}
\newcommand{\frsemCat}{\mathbb{S}_{\dr}}
\newcommand{\fsemCat}{\mathbb{S}}
\newcommand{\meagersymb}{\mathrm{M}}
\newcommand{\mpsemCat}{\mathbb{S}^{\meagersymb,\mathbb{P}}_{\dr}}
\newcommand{\mrsemCat}{\mathbb{S}^{\meagersymb}_{\dr}}

\newcommand{\fpwpFunctor}{\mathcal{W}^{\mathrm{P}}_{\dr}}
\newcommand{\frwpFunctor}{\mathcal{W}_{\dr}}
\newcommand{\fwpFunctor}{\mathcal{W}}
\newcommand{\mpwpFunctor}{\mathcal{W}^{\meagersymb,\mathrm{P}}_{\dr}}
\newcommand{\mrwpFunctor}{\mathcal{W}^{\meagersymb}_{\dr}}

\newcommand{\minfunc}{\mathcal{P}^{\mathrm{min}}} 
\newcommand{\maxfunc}{\mathcal{P}^{\mathrm{max}}} 
\newcommand{\uncompfunc}{\mathcal{M}}
\newcommand{\kleisli}[2]{\mathit{K}\!\ell(#2)}
\newcommand{\swapsyb}{\sigma}
\newcommand{\swap}[2]{\swapsyb_{#1, #2}}

\newcommand{\id}{\mathrm{id}}

\newcommand{\powerfuncf}{\mathcal{P}_{\mathrm{f}}} 

\renewcommand{\cref}[1]{\Cref{#1}}
\creflabelformat{enumi}{#2(#1)#3}
\crefname{theorem}{Thm.}{Theorems}
\crefname{definition}{Def.}{Defs}
\crefname{proposition}{Prop.}{Props}
\crefname{lemma}{Lem.}{Lemmas}
\crefname{proof}{Proof.}{Proofs}
\crefname{appendix}{Appendix}{Appendixes}
\crefformat{section}{{\S}#2#1#3}
\crefname{figure}{Fig.}{Figs}
\Crefname{equation}{}{}

\usepackage[all]{xy} \SelectTips{cm}{}  
\newdir{ >}{{}*!/-8pt/@{>}}  
\newdir{|>}{%
	!/1.6pt/@{|}*:(1,-.2)@^{>}*:(1,+.2)@_{>}}

\newcommand{\myparagraph}[1]{\vspace{.3em}\noindent\textbf{#1}\;}

\newcommand{\myvspin}[1][-0.3em]{\vspace{#1}} 
\newcommand{\myvspaf}[1][-0.4em]{\vspace{#1}}
\newcommand{\myvsplistaf}[1][-0.4em]{\vspace{#1}}
\newcommand{\myvspmathbef}[1][-1.em]{\\[#1]} 
\newcommand{\myvspmathaf}[1][.5em]{\\[#1]} 
\newenvironment{myminipage}{\begin{minipage}{\textwidth}}{\end{minipage}}

\newcommand{\myvspfigcapbef}[1][-1em]{\vspace{#1}}

\newcommand{\omitatsubm}[1]{} 

\title{Compositional Solution of Mean Payoff Games by String Diagrams
} 

\titlerunning{Compositional Solution of Mean Payoff Games by String Diagrams
} 

\author{Kazuki Watanabe}{National Institute of Informatics and SOKENDAI, Tokyo, Japan \and \url{https://group-mmm.org/~kazuki/} }{kazukiwatanabe@nii.ac.jp}{https://orcid.org/0000-0002-4167-3370}{}

\author{Clovis Eberhart}{
National Institute of Informatics and JFLI (IRL 3527, CNRS), Tokyo, Japan
}{eberhart@nii.ac.jp}{https://orcid.org/0000-0003-3009-6747}{}

\author{Kazuyuki Asada}{Tohoku University, Japan \and \url{https://www.riec.tohoku.ac.jp/~asada/} }{asada@riec.tohoku.ac.jp}{https://orcid.org/0000-0001-8782-2119}{}

\author{Ichiro Hasuo}{
National Institute of Informatics and SOKENDAI, Tokyo, Japan
}{i.hasuo@acm.org}{https://orcid.org/0000-0002-8300-4650}{}

\authorrunning{K. Watanabe, C. Eberhart, K. Asada, and I. Hasuo} 

\Copyright{Kazuki Watanabe, Clovis Eberhart, Kazuyuki Asada, and Ichiro Hasuo} 

\ccsdesc[100]{Theory of computation~Verification by model checking} 

\keywords{compositionality, verification, mean payoff game, category theory, monoidal category, string diagram, traced monoidal category, compact closed category} 

\category{} 

\relatedversion{} 

\nolinenumbers 

\EventEditors{John Q. Open and Joan R. Access}
\EventNoEds{2}
\EventLongTitle{42nd Conference on Very Important Topics (CVIT 2016)}
\EventShortTitle{CVIT 2016}
\EventAcronym{CVIT}
\EventYear{2016}
\EventDate{December 24--27, 2016}
\EventLocation{Little Whinging, United Kingdom}
\EventLogo{}
\SeriesVolume{42}
\ArticleNo{23}

\begin{document}

\maketitle

\begin{abstract}
Following our recent development of a compositional model checking algorithm for Markov decision processes, we present a compositional framework for solving mean payoff games (MPGs). The framework is derived from category theory, specifically that of monoidal categories: MPGs (extended with open ends) get composed in so-called string diagrams and thus organized in a monoidal category; their solution is then expressed as a functor, whose preservation properties embody compositionality. As usual, the key question to compositionality is how to enrich the semantic domain; the categorical framework gives an informed guidance in solving the question by singling out the algebraic structure required in the extended semantic domain. We implemented our compositional solution in Haskell; depending on benchmarks, it can outperform an existing algorithm by an order of magnitude.

\end{abstract}

\section{Introduction}
The current paper is the latest result in our pursuit~\cite{Watanabe21,Watanabe23} of compositional algorithms for model checking and game solving. We have studied solution of parity games~\cite{Watanabe21} and model checking of Markov decision processes (MDPs)~\cite{Watanabe23}; the latter successfully yielded an algorithm and an efficient implementation. Our approach features a structural and algebraic theory in the language of \emph{monoidal categories}~\cite{MacLane2}, where target systems/games are composed in the graphical language of \emph{string diagrams}.
In this paper, we demonstrate the power of the categorical approach by exhibiting another target problem, namely solution of \emph{mean payoff games (MPGs)}.

MPGs have been extensively studied for its application to formal verification of \emph{quantitative} systems.
MPGs play important roles in the modelling of embedded systems~\cite{ChakrabartiAHS03}, quantitative LTL synthesis~\cite{TomitaUSHY17}, and temporal networks~\cite{CominR15,CominPR17}. 
Therefore, an efficient algorithm for solving MPGs is much desired.

Recent studies~\cite{Brim11,Benerecetti20} present pseudo-polynomial algorithms for solving MPGs that use \emph{progress measures}~\cite{jurdzinski2000small} as a key ingredient.
Benerecetti {\it et al.}~\cite{Benerecetti20} exploit the notion of \emph{quasi dominion}~\cite{benerecetti2018solving} and experimentally show that their algorithm is remarkably faster than the  algorithm in~\cite{Brim11} that is conceptually simpler.

Although recent work has made great progress in the search for efficient algorithms for MPGs, there have been no algorithms with compositionality, a property with  both mathematical blessings and a proven record of success. A compositional algorithm
is  a \emph{divide-and-conquer} method, where a large system is divided  into smaller components and 
the results are combined to analyze the original whole system. 
Compositionality in formal verification has been pursued in~\cite{ClarkeLM89,KwiatkowskaNPQ13,Stephens15,Watanabe21,Watanabe23}.

\begin{figure}[tb]
\scalebox{0.5}{
  \begin{minipage}[]{40em}
  \vspace{-8em}
    \begin{tikzpicture}[
              innode/.style={draw, rectangle, minimum size=1cm},
              innodemini/.style={draw, rectangle, minimum size=0.5cm},
              interface/.style={inner sep=0},
              innodeeve/.style={draw, circle, minimum size=0.2cm},
              innodeadam/.style={draw, rectangle, minimum size=0.2cm},
              ]
              \node[interface] (rdo1) at (0cm, 0.25cm) {};
              \node[interface] (rdo2) at (0cm, 0cm) {};
              \node[interface] (rdo3) at (0cm, -0.25cm) {};
               \node[interface] (rcdo1) at (2cm, 0.25cm) {};
              \node[interface] (rcdo2) at (2cm, 0cm) {};
              \node[interface] (rcdo3) at (2cm, -0.25cm) {};
              \node[innode] (game1) at (1cm, 0cm) {\scalebox{2}{$\mpg{A}$}};
              \draw[-Stealth] (rdo1) to ($(game1.north west)!0.5!(game1.west)$);
              \draw[-Stealth] (rdo2) to (game1);
              \draw[Stealth-] (rdo3) to ($(game1.south west)!0.5!(game1.west)$);
             \draw[-Stealth] ($(game1.north east)!0.5!(game1.east)$) to (rcdo1);
             \draw[Stealth-] (game1) to (rcdo2);
             \draw[Stealth-] ($(game1.south east)!0.5!(game1.east)$) to (rcdo3);
             \node[interface] (seqcomp) at (2.5cm, 0cm) {\scalebox{2}{$\seqcomp$}};
             \node[interface] (rdo1b) at (3cm, 0.25cm) {};
             \node[interface] (rdo2b) at (3cm, 0cm) {};
             \node[interface] (rdo3b) at (3cm, -0.25cm) {};
             \node[interface] (rcdo1b) at (5cm, 0.25cm) {};
             \node[interface] (rcdo2b) at (5cm, -0.25cm) {};
             \node[innode] (game2) at (4cm, 0cm) {\scalebox{2}{$\mpg{B}$}};
              \draw[-Stealth] (rdo1b) to ($(game2.north west)!0.5!(game2.west)$);
              \draw[Stealth-] (rdo2b) to (game2);
              \draw[Stealth-] (rdo3b) to ($(game2.south west)!0.5!(game2.west)$);
             \draw[-Stealth] ($(game2.north east)!0.5!(game2.east)$) to (rcdo1b);
             \draw[Stealth-] ($(game2.south east)!0.5!(game2.east)$) to (rcdo2b);

             \node[interface] (equal) at (6.5cm, 0cm) {\scalebox{3}{$=$}};

             \node[interface] (rdo1ab) at (8cm, 0.25cm) {};
             \node[interface] (rdo2ab) at (8cm, 0cm) {};
             \node[interface] (rdo3ab) at (8cm, -0.25cm) {};
             \node[innode] (game1ab) at (9cm, 0cm) {\scalebox{2}{$\mpg{A}$}};
             \node[innode] (game2ab) at (11cm, 0cm) {\scalebox{2}{$\mpg{B}$}};
             \node[interface] (rcdo1ab) at (12cm, 0.25cm) {};
             \node[interface] (rcdo2ab) at (12cm, -0.25cm) {};
             \draw[-Stealth] (rdo1ab) to ($(game1ab.north west)!0.5!(game1ab.west)$);
             \draw[-Stealth] (rdo2ab) to (game1ab);
             \draw[Stealth-] (rdo3ab) to ($(game1ab.south west)!0.5!
             (game1ab.west)$);
             \draw[-Stealth] ($(game1ab.north east)!0.5!
             (game1ab.east)$) to ($(game2ab.north west)!0.5!
             (game2ab.west)$);
             \draw[Stealth-] (game1ab) to (game2ab);
             \draw[Stealth-] ($(game1ab.south east)!0.5!
             (game1ab.east)$) to ($(game2ab.south west)!0.5!
             (game2ab.west)$);
             \draw[-Stealth] ($(game2ab.north east)!0.5!(game2ab.east)$) to (rcdo1ab);
             \draw[Stealth-] ($(game2ab.south east)!0.5!(game2ab.east)$) to (rcdo2ab);
             \node[interface] (comma) at (12.3cm, -0.4cm) {\scalebox{2}{$,$}};
          \end{tikzpicture}
    \end{minipage}

    \begin{minipage}[t]{50em}
    \begin{tikzpicture}[
              innode/.style={draw, rectangle, minimum size=1cm},
              innodemini/.style={draw, rectangle, minimum size=0.5cm},
              interface/.style={inner sep=0},
              innodeeve/.style={draw, circle, minimum size=0.2cm},
              innodeadam/.style={draw, rectangle, minimum size=0.2cm},
              ]
              \node[interface] (rdo1) at (0cm, 0.25cm) {};
              \node[interface] (rdo2) at (0cm, 0cm) {};
              \node[interface] (rdo3) at (0cm, -0.25cm) {};
               \node[interface] (rcdo1) at (2cm, 0.25cm) {};
              \node[interface] (rcdo2) at (2cm, 0cm) {};
              \node[interface] (rcdo3) at (2cm, -0.25cm) {};
              \node[innode] (game1) at (1cm, 0cm) {\scalebox{2}{$\mpg{A}$}};
              \draw[-Stealth] (rdo1) to ($(game1.north west)!0.5!(game1.west)$);
              \draw[-Stealth] (rdo2) to (game1);
              \draw[Stealth-] (rdo3) to ($(game1.south west)!0.5!(game1.west)$);
             \draw[-Stealth] ($(game1.north east)!0.5!(game1.east)$) to (rcdo1);
             \draw[Stealth-] (game1) to (rcdo2);
             \draw[Stealth-] ($(game1.south east)!0.5!(game1.east)$) to (rcdo3);
             \node[interface] (sum) at (2.5cm, 0cm) {\scalebox{2}{$\oplus$}};
             \node[interface] (rdo1b) at (3cm, 0.25cm) {};
             \node[interface] (rdo2b) at (3cm, 0cm) {};
             \node[interface] (rdo3b) at (3cm, -0.25cm) {};
             \node[interface] (rcdo1b) at (5cm, 0.25cm) {};
             \node[interface] (rcdo2b) at (5cm, -0.25cm) {};
             \node[innode] (game2) at (4cm, 0cm) {\scalebox{2}{$\mpg{B}$}};
              \draw[-Stealth] (rdo1b) to ($(game2.north west)!0.5!(game2.west)$);
              \draw[Stealth-] (rdo2b) to (game2);
              \draw[Stealth-] (rdo3b) to ($(game2.south west)!0.5!(game2.west)$);
             \draw[-Stealth] ($(game2.north east)!0.5!(game2.east)$) to (rcdo1b);
             \draw[Stealth-] ($(game2.south east)!0.5!(game2.east)$) to (rcdo2b);

             \node[interface] (equal) at (6.5cm, 0cm) {\scalebox{3}{$=$}};

             \node[interface] (rdo1aba) at (8cm, 1.25cm) {};
             \node[interface] (rdo2aba) at (8cm, 1cm) {};
             \node[interface] (rdo3aba) at (8cm, 0.75cm) {};
             \node[interface] (rdo1abb) at (8cm, -0.75cm) {};
             \node[interface] (rdo2abb) at (8cm, -1cm) {};
             \node[interface] (rdo3abb) at (8cm, -1.25cm) {};
             \node[interface] (rcdo1aba) at (12cm, 1.25cm) {};
             \node[interface] (rcdo2aba) at (12cm, 1cm) {};
             \node[interface] (rcdo3aba) at (12cm, 0.75cm) {};
             \node[interface] (rcdo1abb) at (12cm, -0.75cm) {};
             \node[interface] (rcdo2abb) at (12cm, -1cm) {};
             \node[interface] (rcdo3abb) at (12cm, -1.25cm) {};
             \node[innode] (game1ab) at (10cm, 1cm) {\scalebox{2}{$\mpg{A}$}};
             \node[innode] (game2ab) at (10cm, -1cm) {\scalebox{2}{$\mpg{B}$}};
             \draw[-Stealth] (rdo1aba) to ($(game1ab.north west)!0.5!(game1ab.west)$);
             \draw[-Stealth] (rdo2aba) to (game1ab);
             \draw[-Stealth] (rdo3aba) -- (8.5cm, 0.75cm) -- (9cm, -0.75cm) -- ($(game2ab.north west)!0.5!(game2ab.west)$);
             \draw[Stealth-] (rdo1abb) -- (8.5cm, -0.75cm) -- (9cm, -1cm) -- (game2ab);
             \draw[Stealth-] (rdo2abb) -- (8.5cm, -1cm) -- (9cm, -1.25cm) -- ($(game2ab.south west)!0.5!(game2ab.west)$);
             \draw[Stealth-] (rdo3abb) -- (8.5cm, -1.25cm) -- (9cm, 0.75cm) -- ($(game1ab.south west)!0.5!(game1ab.west)$);
             \draw[-Stealth] ($(game1ab.north east)!0.5!(game1ab.east)$) to (rcdo1aba);
             \draw[-Stealth] ($(game2ab.north east)!0.5!(game2ab.east)$) -- (11cm, -0.75cm) -- (11.5cm, 0.75cm) -- (rcdo3aba);
             \draw[Stealth-] ($(game2ab.south east)!0.5!(game2ab.east)$) -- (11cm, -1.25cm) -- (11.5cm, -0.75cm) -- (rcdo1abb);
             \draw[Stealth-] (game1ab) -- (11.1cm, 1cm) -- (11.6cm, -1cm) -- (rcdo2abb);
             \draw[Stealth-] ($(game1ab.south east)!0.5!(game1ab.east)$) -- (11cm, 0.75cm) -- (11.5cm, -1.25cm) -- (rcdo3abb);
          \end{tikzpicture}
    \end{minipage}

    }
    \myvspfigcapbef[-2em]
\caption{Sequential composition $\seqcomp$, and sum $\oplus$ of MDPs, illustrated.}
\label{fig:seqCompOplusIllustrated}
\end{figure}


In this paper, we present a novel compositional algorithm for solving MPGs; it is  a mapping $\fwpFunctor$ from ``compositional MPGs'' to ``solutions'' (such as win/lose).
This mapping must preserve a certain algebraic structure that offers operations for composing MPGs.
Following \cite{Watanabe21,Watanabe23},
we identify the relevant algebraic structure as that of \emph{compact closed categories (compCC)}, where one can use the graphical calculus of \emph{string diagrams} to compose MPGs. In string diagrams, MPGs---extended with \emph{open ends} and called \emph{open MPGs}---can be composed using two binary operations (\emph{sequential composition} $\seqcomp$ and \emph{sum} $\oplus$). See \cref{fig:seqCompOplusIllustrated}; one can see that string diagrams for compCC are a natural calculus for not only MPGs but also graph-based systems in general such as MDPs and parity games. 


We organize open MPGs as arrows of a category $\oMPG$ (its objects are given by a suitable notion of arity). Then it is natural to seek 1) a solution domain $\fsemCat$ that has a compCC structure (thus called a \emph{semantic category}), and 2) a solution map $\fwpFunctor\colon \oMPG\to \fsemCat$ that preserves compCC structures. Such a structure-preserving map $\fwpFunctor$ between compCC is called a \emph{compact closed functor}~\cite{KELLY1980193}. Its preservation properties can be spelled out as
\myvspmathbef[-.9em] 
\begin{myminipage}
\begin{equation*}
    \fwpFunctor(\mpg{A}\seqcomp \mpg{B}) = \fwpFunctor(\mpg{A})\seqcomp \fwpFunctor(\mpg{B}), \ \  \fwpFunctor(\mpg{A}\oplus \mpg{B}) = \fwpFunctor(\mpg{A})\oplus \fwpFunctor(\mpg{B})
\end{equation*}
\end{myminipage}
\myvspmathaf
which embodies desired compositionality. 
This functor $\fwpFunctor$---we call it a \emph{winning-position functor}---compositionally computes all the winning (initial) positions of a given open MPG.


\begin{figure}[tbp]\centering
\begin{math}\mspace{-50mu}
\vcenter{\xymatrix@C+1.6em@R=0.1em{
  **[r]{\footnotesize\begin{array}[c]{l}
	    \text{bidirectional, MPGs}\qquad\qquad\qquad\qquad\qquad\\
	\text{(compact closed)}\end{array}
   }
  &
  {\mspace{20mu}\mathllap{\oMPG\, \defeq\;} 
  \Int(\roMPG)}
     \ar[r]^-{\fwpFunctor\defeq\Int(\frwpFunctor)}
   &
  {\Int(\frsemCat)\mathrlap{\;=:\,\fsemCat }}
  \\
  **[r]{\footnotesize\begin{array}[c]{l}\footnotesize
   \text{unidirectional, MPGs}\\
   \footnotesize
   \text{(traced monoidal)}
\end{array}
   }
   \ar@(lu,ld)[u]^{
    \begin{array}{r}
     \text{\small the Int}\\[-4pt]
     \text{\small constr.}
    \end{array}}
  &
  {\roMPG}
      \ar[r]_-{\frwpFunctor }
  &
  {\frsemCat}
  \\
  **[r]{\footnotesize\begin{array}[c]{l}\footnotesize
   \text{unidirectional, plays}\\
   \footnotesize
   \text{(traced monoidal)}
\end{array}
   }
   \ar@(lu,ld)[u]^{
    \begin{array}{r}
     \text{\small change}\\[-4pt]
     \text{\small of base}
    \end{array}}
  &
  {\roPlay}
      \ar[r]_-{\fpwpFunctor}
  &
  {\fpsemCat}
 }}
\end{math}
\myvspfigcapbef
\caption{Categories of MPGs/Plays, semantic categories, and winning-position functors.}
\label{fig:catsFunctors}
\end{figure}
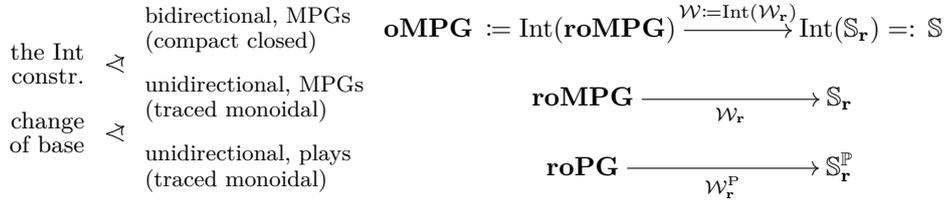

\begin{wrapfigure}[4]{R}{0.35\textwidth}
\centering
\vspace{-1.5em}
  \begin{minipage}[t]{0.35\textwidth}
  \centering
  \scalebox{0.7}{
  \begin{tikzpicture}[
              innode/.style={draw, rectangle, minimum size=1cm},
              interface/.style={inner sep=0},
              innodemini/.style={draw, circle, minimum size=0.2cm}
              ]
              \node[interface] (rdo1) at (0cm, 0.3cm) {};
              \node[interface] (rdo2) at (0cm, 0.1cm) {};
              \node[interface] (rdo3) at (0cm, -0.1cm) {};
              \node[interface] (rdo4) at (0cm, -0.3cm) {};
              \node[innode] (game1) at (1cm, 0cm) {$\mpg{A}$};
              \draw[-Stealth] (rdo1) to ($(game1.north west)!0.4!(game1.west)$);
              \draw[-Stealth] (rdo2) to ($(game1.north west)!0.8!(game1.west)$);
              \draw[-Stealth] (rdo3) to ($(game1.south west)!0.8!(game1.west)$);
              \draw[-Stealth] (rdo4) to ($(game1.south west)!0.4!(game1.west)$);
              \node[interface] (rcdo1) at (2cm, 0.3cm) {};
              \node[interface] (rcdo2) at (2cm, 0.1cm) {};
              \node[interface] (rcdo3) at (2cm, -0.1cm) {};
              \draw[Stealth-] (rcdo1) to ($(game1.north east)!0.4!(game1.east)$);
              \draw[Stealth-] (rcdo2) to ($(game1.north east)!0.8!(game1.east)$);
              \draw[Stealth-] (rcdo3) to ($(game1.south east)!0.8!(game1.east)$);
               \node[interface] (sum) at (3cm, 0.3cm) {\scalebox{2}{$\overset{\trsyb_{1;3,2}}{\longmapsto}$}};
                \node[innode] (game2) at (5cm, 0cm) {$\mpg{A}$};
              \node[interface] (rdo21) at (4cm, 0.1cm) {};
              \node[interface] (rdo22) at (4cm, -0.1cm) {};
              \node[interface] (rdo23) at (4cm, -0.3cm) {};
               \node[interface] (rcdo21) at (6cm, 0.1cm) {};
              \node[interface] (rcdo22) at (6cm, -0.1cm) {};
              \draw[-Stealth] (4.5cm, 1cm) arc [radius=0.3, start angle = 90, end angle=270];
              \draw[-Stealth] (rdo21) to ($(game2.north west)!0.8!(game2.west)$);
              \draw[-Stealth] (rdo22) to ($(game2.south west)!0.8!(game2.west)$);
              \draw[-Stealth] (rdo23) to ($(game2.south west)!0.4!(game2.west)$);
              \draw[-Stealth] (5.5cm, 0.4cm) arc [radius=0.3, start angle = 270, end angle=450];
              \draw[Stealth-] (rcdo21) to ($(game2.north east)!0.8!(game2.east)$);
              \draw[Stealth-] (rcdo22) to ($(game2.south east)!0.8!(game2.east)$);
              \draw[-] (4.5cm, 1cm) to (5.5cm, 1cm);
          \end{tikzpicture}
    }
    \end{minipage}
    \caption{Trace operator.}
    \label{fig:trace_operator}
\end{wrapfigure}

To obtain a suitable semantic category $\fsemCat$ and a winning-position functor $\fwpFunctor\colon  \oMPG\to \fsemCat$, we follow the categorical workflow introduced in~\cite{Watanabe23}, in which  $\fsemCat$ and $\fwpFunctor$ are obtained in a two-step process. See~\cref{fig:catsFunctors}.
In the two steps we utilize general categorical constructions, namely the \emph{Int construction}~\cite{joyal1996} and the \emph{change of base} construction~\cite{Eilenberg65,cruttwell2008normed}. 
The Int construction turns a framework of ``unidirectional'' open MPGs into that of  ``bidirectional'' open MPGs,
and for this we need a \emph{trace operator} for the category $\roMPG$ of ``unidirectional'' open MPGs.
Intuitively, the trace operator is an algebraic operation that creates a loop in an MPG (\cref{fig:trace_operator}).
In the other step, the change of base construction builds the category $\roMPG$  of ``unidirectional'' open MPGs from a category $\roPlay$ of ``plays,'' by adding two nondeterministic structures for the two players $\eve$ and $\adam$. 
  Change of base has been applied in computer science to \emph{game semantics}~\cite{laird2017qualitative}, too.

The two-step process described above (\cref{fig:catsFunctors}) on the syntax side (i.e.\ $\roPlay\mapsto \roMPG\mapsto \oMPG$) also takes place in parallel on the semantics side, i.e., on semantic categories and winning-position functors.
The change of base construction (does not generally, but does in our specific case)  lift a trace operator, so we need a trace operator for the semantic category $\fpsemCat$ for ``plays.'' For its construction, we use a \emph{priority-based} technique inspired by the construction for parity games in~\cite{Watanabe21} (which is further inspired by~\cite{grellois2015finitary}).

We implemented the winning-position functor $\fwpFunctor\colon \oMPG\to\fsemCat$;
it receives a string diagram of oMPGs and outputs the solution of its composition. 
 Experiments show that our implementation (we call it $\compMPG$)
outperforms the known algorithm QDPM~\cite{Benerecetti20} on both 1) simple but realistic benchmarks and 2) randomized benchmarks. $\compMPG$ solved benchmarks as big as $10^7$ positions within three seconds, demonstrating its efficiency.


Our contributions are summerized as follows:
\begin{enumerate}
    \item A compositional algorithm for solving MPGs composed by string diagrams.
 \item Its structural and disciplined derivation by category theory.
    \item Its implementation $\compMPG$ and experimental evaluation that shows its efficiency.
\end{enumerate}
\myvsplistaf



\myparagraph{Related Work}
We have already mentioned related work on pseudo-polynomial algorithms for MPGs~\cite{Brim11,Benerecetti20}, and compositional model checking~\cite{ClarkeLM89,KwiatkowskaNPQ13, Watanabe21,Watanabe23,Stephens15,DBLP:conf/csl/TsukadaO14}. Here, we give a detailed comparison between our work and closely related work~\cite{Stephens15,Watanabe21,DBLP:conf/csl/TsukadaO14,Watanabe23}.

The work~\cite{Stephens15,RathkeSS14} uses string diagrams to compose Petri nets and computes
 reachability probabilities  in a compositional way. A major difference from our work is that they do not allow loops in composition; consequently, they use symmetric monoidal categories without traces or compact closed structures. We find loops to be essential in accommodating real-world examples (see e.g.\ \cref{fig:exbenchmark}).  Treatment of loops is a major theoretical challenge, too, which we successfully address by a priority-based construction of traces in $\fpsemCat$.


The formalism of string diagrams
originates  from the theory of \emph{monoidal categories} (see e.g.~\cite[Chap.~XI]{MacLane2}). Capturing the mathematical essence of the algebraic structure of arrow composition $\circ$ and tensor product $\otimes$---they correspond to $;$ and $\oplus$ in this work, respectively---monoidal categories and string diagrams have found their application in a vast variety of scientific disciplines, such as  quantum field theory~\cite{khovanov2002functor}, quantum mechanics and computation~\cite{heunen2019categories}, linguistics~\cite{PiedeleuKCS15}, signal flow diagrams~\cite{BonchiHPSZ19}, and so on.

A compositional framework for parity games with string diagrams is introduced in~\cite{Watanabe21}, whose semantic category is designed following the work~\cite{grellois2015finitary} on denotational semantics of higher-order model checking. Unlike the present paper, the winning-position functor in~\cite{Watanabe21} does not use the change of base construction; as a result, the functor collects too many strategies and is thus not suited for efficient implementation. We note that parity games can be reduced to MPGs~\cite{Jurdzinski98} and solved by our current algorithm; besides, we expect it is possible  to adapt our current algorithm from MPGs to parity games.

A compositional algorithm for computing expected rewards of MDPs with
string diagrams is introduced in~\cite{Watanabe23}. 
In the current paper, we follow the categorical workflow in~\cite{Watanabe23}---in fact, \cref{fig:catsFunctors} is very similar to one in~\cite{Watanabe23}. However, there is a technical difference in the most challenging part of the workflow, namely the construction of a trace operator in $\fpsemCat$ (corresponding to $\mathbb{S}^{\mathrm{MC}}_{\mathbf{r}}$ in~\cite{Watanabe23}).  The construction in~\cite{Watanabe23} is least fixed point-based: it captures an arbitrary number of iterations and the reward collected in its course. In constrast, our current construction for MPGs is priority-based: we think of MPGs as an extension of parity games from finitely many priorities to  infinitely many; therefore we adapt the construction of traces in~\cite{Watanabe21,grellois2015finitary}. 

\myparagraph{Organization}
In~\cref{sec:graphOMPGs}, we introduce open MPGs and their  semantics of the conventional (non-categorical) style, and define a compact closed category $\oMPG$ of open MPGs (as well as $\roMPG$).
In~\cref{sec:decEq}, we define the category $\roPlay$ of ``plays'' (more precisely \emph{rightward open play graphs}) and their decomposition equalities, which are key properties towards compositionality and are explained without using category theory.
In~\cref{sec:semanticCategories}, we define all the semantic structures in \cref{fig:catsFunctors},
and show the main theorem (compositionality of  $\fwpFunctor$).
In~\cref{sec:impAndExp}, we show the results of experiments and address some research questions.
\myparagraph{Notations}
For natural numbers $m$ and $n$, we let $\nset{m, n} \defeq \{m, m+1,\dots, n-1, n\}$; as a special case, we let $\nset{m} \defeq \{1, 2, \dots, m\}$ (we let $\nset{0} = \emptyset$ by convention).
 $X + Y$ denotes the disjoint union of sets $X, Y$. 
For a category $\mathcal{C}$ and its objects $X$ and $Y$, we write $\mathcal{C}(X,Y)$ for the set of arrows from $X$ to $Y$.

\section{Compact Closed Category of Open Mean Payoff Games}
\label{sec:graphOMPGs}
We introduce \emph{open mean payoff games} (\emph{oMPGs}), an extension of MPGs with open ends, in \cref{subsec:oMPGs}, where
we also give their semantics in a conventional style.
Then, for our compositional framework, we introduce a compact closed category (compCC) $\oMPG$ of oMPGs.
As explained in the introduction (see \cref{fig:catsFunctors}), for technical convenience
we define $\oMPG$ by the $\Int$ construction~\cite{joyal1996}, which is given in \cref{sec:compCCofOMPGbyInt}.
This takes as an input a traced symmetric monoidal category (TSMC), so in \cref{sec:TSMCofROMPG} we define a TSMC $\roMPG$ of ``unidirectional'' open MPGs, which we call \emph{rightward open MPGs}. See~\cite{joyal1996} for details.

\subsection{Open Mean Payoff Games}
\label{subsec:oMPGs}
\begin{wrapfigure}[4]{r}{0pt}
\scalebox{0.7}{
  \begin{minipage}[t]{15em}
  \vspace{-\baselineskip}
    \begin{tikzpicture}[
              innode/.style={draw, rectangle, minimum size=1cm},
              innodemini/.style={draw, rectangle, minimum size=0.5cm},
              interface/.style={inner sep=0},
              innodeeve/.style={draw, circle, minimum size=0.2cm},
              innodeadam/.style={draw, rectangle, minimum size=0.2cm},
              ]
              \node[innodeeve] (game1) at (1cm, 0cm) {$3.1$};
              \node[innodeadam] (game2) at (2.5cm, 0.25cm) {$-4.5$};
              \node[innodeadam] (game3) at (4cm, 0cm) {$2$};
              \node[interface] (rcdo1) at (5cm, 1.25cm) {\;$1_{\dr}$};
              \node[interface] (rcdo2) at (5cm, 0.25cm) {\;$2_{\dr}$};
              \node[interface] (rcdo3) at (5cm, -0.25cm) {\;$1_{\dl}$};
              \node[interface] (rdo1) at (0cm, 1.25cm) {$1_{\dr}$\;};
              \node[interface] (rdo2) at (0cm, 0.75cm) {$2_{\dr}$\;};
              \node[interface] (rdo3) at (0cm, 0.25cm) {$3_{\dr}$\;};
              \node[interface] (rdo4) at (0cm, -0.25cm) {$1_{\dl}$\;};
              \draw[->] (rdo1) to (rcdo1);
              \draw[->] (rdo2) to (game2);
              \draw[->] (rdo3) to (game1);
              \draw[<-] (rdo4) to (game1);
              \draw[->] (game1) to (game2);
              \draw[->] (game2) to (game3);
              \draw[->] ($(game3.south west)!0.3!(game3.west)$) to ($(game1.south east)!0.3!(game1.east)$);
              \draw[->] (game3) to (rcdo2);
              \draw[<-] (game3) to (rcdo3);
          \end{tikzpicture}
    \end{minipage}
    }
    \caption{An open MPG.}
    \label{fig:openMPG}
  \end{wrapfigure}
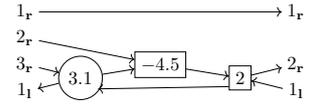

We first define open MPGs. They are connected by \emph{open ends}, which come with a notion of \emph{arity}---the numbers of open ends on their left and right, distinguishing leftward and rightward ones.
 As an example, \cref{fig:openMPG} is an open MPG whose arity on the left is $(3, 1)$, and one on the right is $(2, 1)$.

\begin{definition}[open MPG (oMPG)]
\myvspin[-0.3em]\label{def:oMPG}
 An \emph{open MPG} $\mpg{A}$ (from $\bfn{m}$ to $\bfn{n}$) is a tuple $(\bfn{m},\bfn{n},Q, E,\rho, w)$ of the following data. 
 \begin{enumerate}
    \item $\bfn{m} = (\nfr{m}, \nfl{m})$ and $\bfn{n} = (\nfr{n}, \nfl{n})$ are pairs of natural numbers; they are called the \emph{left-arity} and the \emph{right-arity}, respectively. Moreover, elements of $\nset{\nfr{m} + \nfl{n}}$ are called \emph{entrances}, and those of $\nset{\nfr{n} + \nfl{m}}$ are called \emph{exits}. Entrances and exits are also called \emph{open ends}.
    \item $Q$ is a finite set of \emph{positions}. 
    \item $E\subseteq (\nset{\nfr{m}+\nfl{n}} + Q)\times (\nset{\nfr{n}+\nfl{m}} + Q)$, whose elements are called \emph{edges}. In addition, $E$ must satisfy $(\dagger)$ for any $i\in \nset{\nfr{m}+\nfl{n}}$, there is exactly one $t\in \nset{\nfr{n}+\nfl{m}} + Q$ such that $(i, t)\in E$ (each entrance has a unique successor), and $(\ddagger)$  for any $j\in \nset{\nfr{n}+\nfl{m}}$, there is at most one $s\in \nset{\nfr{m}+\nfl{n}} + Q$ such that $(s, j)\in E$ (each exit has at most one predecessor).
    We denote $(s,t) \in E$ as $s \rightarrow_{\mpg{A}} t$.
    \item $\rho$ is a function $\rho:Q\rightarrow \{\eve, \adam\}$, which assigns a \emph{role} to each position.
    \item $w$ is a function $w:Q\rightarrow \Real$, which assigns a \emph{weight} to each position.
 \end{enumerate}
 
\end{definition}
\myvspaf

Note that the conditions $(\dagger)$ and $(\ddagger)$ are for technical convenience; they can be easily enforced by adding an extra ``access'' position to an entrance or an exit. The condition $(\ddagger)$  will be important for the definition of the mean payoff  condition before \cref{def:denotOfTDP}. 
For an oMPG $\mpg{A}$, we may write the components of the tuple as $\bfn{m}^\mpg{A}$, $\bfn{n}^\mpg{A}$, $Q^\mpg{A}$, $ E^\mpg{A}$, $\rho^\mpg{A}$, and $w^\mpg{A}$.

We will give our compositional (categorical) semantics of oMPGs in \cref{sec:semanticCategories},
but before that we give a non-categorical semantics of oMPGs in a conventional style for MPGs.
Roughly, our problem is to decide, for a  given oMPG $\mpg{A}$, whether an entrance $i$ is \emph{winning}, \emph{losing}, or \emph{pending}, where pending is a limbo status (due to openness) between winning and losing. 

 We start with defining a \emph{play} on an oMPG, which is a  possibly infinite maximal sequence of positions or open ends.

\begin{definition}[play]
\myvspin
Let $\mpg{A}=  (\bfn{m},\bfn{n},Q, E,\rho, w)$ be an oMPG. A (maximal possibly infinite) play $\play = (s_j)_{j\in J}$ in $\mpg{A}$ from an entrance $i\in \nset{\nfr{m} + \nfl{n}}$ is a possibly infinite sequence such that: (i) $ J = \Natm{M} $ for some $M\in \Nat$ or $J = \Nat$, (ii) $s_0 = i$, (iii) if $j+1\in J$, then $(s_j, s_{j+1})\in E$ for each $j\in J$ (iv) if $J = \Natm{M}$, then for any $s\in Q +\nset{\nfr{n}+\nfl{m}}$, $(s_{M}, s)\not\in E$. 
\end{definition}
\myvspaf

Next, we define \emph{$\eve$-strategies} and \emph{$\adam$-strategies} on an oMPG, which are partial functions from $\eve$'s positions and $\adam$'s positions to their successor positions (or exits), respectively. We restrict strategies to \emph{memoryless} ones due to the existence of optimal memoryless strategies on MPGs~\cite{ehrenfeucht1979positional}. 
Since our target problem is to decide the winner on MPGs consisting of oMPGs as components, memoryless strategies of oMPGs are sufficient.

\begin{definition}[$\eve$-strategy and $\adam$-strategy]\label{def:strategy}
\myvspin
Let $\mpg{A}=  (\bfn{m},\bfn{n},Q, E,\rho, w)$ be an oMPG. A (memoryless) $\eve$-strategy on $\mpg{A}$ is a partial function $\tau:\rho^{-1}(\eve)\rightharpoonup \nset{\nfr{n}+\nfl{m}} + Q$ such that (i) if $\tau(s_k)$ is defined,  $(s_k, \tau(s_k))\in E$, and (ii) if $\tau(s_k)$ is undefined,  for all $s\in \nset{\nfr{n}+\nfl{m}} + Q$, $(s_k, s)\not\in E$. A $\adam$-strategy on $\mpg{A}$ is defined similarly, by replacing the occurrence of $\eve$ with $\adam$ in the definition. The sets of $\eve$-strategies and $\adam$-strategies on $\mpg{A}$ are denoted by $\stra{\eve}{\mpg{A}}$ and $\stra{\adam}{\mpg{A}}$, respectively.
\end{definition}
\myvspaf

Given an $\eve$-strategy $\tau_{\eve}$ and a $\adam$-strategy $\tau_{\adam}$, for each entrance $i$,
the pair of $\tau_{\eve}$ and $\tau_{\adam}$  induces a  play $\indplay{\tau_{\eve}}{\tau_{\adam}}{i}$ from $i$:

\begin{definition}[play $\indplay{\tau_{\eve}}{\tau_{\adam}}{i}$ induced by strategies; memoryless play]\label{def:indPlayMemlessPlay}
\myvspin
Let $\mpg{A}=  (\bfn{m},\bfn{n},Q, E,\rho, w)$ be an oMPG. The \emph{play} $\indplay{\tau_{\eve}}{\tau_{\adam}}{i}$ \emph{induced} by an $\eve$-strategy $\tau_{\eve}$ and a $\adam$-strategy $\tau_{\adam}$ from an entrance $i\in \nset{\nfr{m}+\nfl{n}}$ is the (unique) play $(s_j)_{j\in J}$ from $i$ such that: 
(i) for any $j\in I$, if $\rho(s_j) = \eve$ and $\tau_{\eve}(s_j)$ is defined, then $s_{j+1} = \tau_{\eve}(s_j)$, and similarly 
(ii) for any $j\in I$, if $\rho(s_j) = \adam$ and $\tau_{\adam}(s_j)$ is defined, then $s_{j+1} = \tau_{\adam}(s_j)$.

We say a play $\play = (s_j)_{j\in J}$ is \emph{memoryless} if  $\play=\indplay{\tau_{\eve}}{\tau_{\adam}}{}$, i.e.\ if it is induced by strategies (that must be memoryless by \cref{def:strategy}). 
\end{definition}
\myvspaf

The \emph{mean payoff condition} (\emph{MP condition}) for infinite plays on oMPGs is defined in the same way as the conventional one on MPGs.

\begin{definition}[MP condition]
\myvspin
\label{def:mpcondplays}
Let $\mpg{A}=  (\bfn{m},\bfn{n},Q, E,\rho, w)$ be an oMPG. We say that an infinite play $\pi = (s_j)_{j\in \Nat}$ in $\mpg{A}$ satisfies the \emph{MP condition} if it satisfies the following inequality:
\myvspmathbef
\begin{myminipage}
\begin{align}
\label{eq:mean_payoff_condition}
    \textstyle\liminf_{n\rightarrow \infty} \tfrac{1}{n} \sum_{j=1}^{n} w(s_j) \geq 0.
\end{align}
\end{myminipage}
\myvspmathaf
\end{definition}
\myvspaf[-1.2em]

The following lemma is fundamental for MPGs. It says that an infinite play---we have to restrict to memoryless ones (\cref{def:indPlayMemlessPlay})---satisfies the MP condition if and only if the sum of the weights of its cycle is non-negative
(note that the cycle is unique since the play is memoryless).
The proof is by elementary calculation.
\begin{lemma}[\cite{ehrenfeucht1979positional,Brim11}]
\myvspin
\label{lem:periodicity}
Let $\pi^{\tau_{\eve}, \tau_{\adam}} = (s_i)_{i\in \Nat}$ be an infinite memoryless play, induced by (memoryless) strategies $\tau_{\eve}, \tau_{\adam}$, on oMPG $\mpg{A} = (\bfn{m},\bfn{n},Q, E,\rho, w)$. 
There is the least $(j,k) \in \Nat \times \Nat$ (w.r.t. the product order) such that: 
for all $i > j$, $s_i = s_{((i-j)\% k) + j  }$. 
Moreover, 
the following conditions are equivalent:
\vspace{.2em}
\begin{itemize}
\item $\textstyle\liminf_{n\rightarrow \infty} \tfrac{1}{n} \sum_{j=1}^{n} w(s_j) \geq 0$,
\item $\sum_{j < i \leq j + k} w(s_i) \geq 0$.
\qed
\end{itemize}
\end{lemma}
\myvspaf

The last lemma is fundamental in this paper, too.
It suggests us the notion of the denotation of a play (\cref{def:denoPlays}): while the MP condition takes an average by $1/n$~\cref{eq:mean_payoff_condition}, the second case of \cref{def:denoPlays} only sums up weights. It is also used in the proof of the \emph{decomposition equalities} for the trace operator (\cref{prop:deqtrace}).

 We now define the \emph{denotation} of a play, which shows the ``result'' of the play.
Here $\const{\eve}$ and $\const{\adam}$ mean that $\eve$ and $\adam$ win, respectively,
and any other result is ``pending.'' 
We divide plays into four cases: (i) going directly to an exit without visiting positions, (ii) going through positions to an exit, (iii) getting stuck in a position whose role is $\adam$, or satisfying the MP condition, (iv) getting stuck in a position whose role is $\eve$, or not satisfying the MP condition.

\begin{definition}[denotation $\denoplay{\pi}{\mpg{A}}$ of plays]
\myvspin
\label{def:denoPlays}
For a set $X$, we define $\fpmonad(X)\defeq X + \Real\times X + \{\const{\eve},\const{\adam}\}$.
Let $\mpg{A}=  (\bfn{m},\bfn{n},Q, E,\rho, w)$ be an oMPG.
The \emph{denotation} $\denoplay{(s_i)_{i\in I}}{\mpg{A}} \in \fpmonad(\nset{\nfr{n} + \nfl{m}}) $ of a play $(s_i)_{i\in I}$ is defined as
\myvspmathbef
\begin{myminipage}
\begin{alignat*}{2}
    &s_{1}  &&\text{ if }|I| = 2\text{ and } s_{1} \text{ is an exit},\\
    &\big(
    \textstyle\sum_{i=1}^{|I|-2}
    w(s_i), s_{|I|-1}\big)  &&\text{ if }I \text{ is finite}, |I| > 2, s_{|I|-1}\text{ is an exit},\\
    &\const{\eve} &&\text{ if }I \text{ is finite and }\rho(s_{|I|-1}) = \adam \text{, or $(s_i)_{i\in I}$ satisfies the MP condition},\\
    &\const{\adam} &&\text{ if }I \text{ is finite and }\rho(s_{|I|-1}) = \eve \text{, or $(s_i)_{i\in I}$ fails the MP condition}.
\end{alignat*} 
\end{myminipage}
\myvspmathaf
We often omit the subscript $\mpg{A}$ when it is clear from the context.
\end{definition}
\myvspaf

Finally, we define the \emph{denotation of an entrance} $i$ of an oMPG. 
It is the collection, for each $\eve$-strategy $\tau_{\eve}$,
of the results of the induced plays against all $\adam$-strategies.

\begin{definition}[denotation of  an entrance]
\myvspin
\label{def:denOfEntrances}
Let $\mpg{A}=  (\bfn{m},\bfn{n},Q, E,\rho, w)$ be an oMPG, and $i\in \nset{\nfr{m}+\nfl{n}}$. The denotation $\denoentry{i}{\mpg{A}}$ of an entrance $i$ is defined by $\denoentry{i}{\mpg{A}} \defeq  \big\{ \{ \denoplay{\indplay{\tau_{\eve}}{\tau_{\adam}}{i}}{\mpg{A}}
 \in \fpmonad(\nset{\nfr{n} + \nfl{m}})
\ | \ \tau_{\adam}\in \stra{\adam}{\mpg{A}}  \} 
\ \big| \allowbreak\ 
\tau_{\eve}\in \stra{\eve}{\mpg{A}}\big\}
$. 
The entrance $i$ is \emph{winning} if $\{ \const{\eve}\}\in \denoentry{i}{\mpg{A}}$, 
and \emph{losing} if $\const{\adam}\in S$ for all $S\in \denoentry{i}{\mpg{A}}$. 
Otherwise, the entrance $i$ is called \emph{pending}.
\end{definition}
\myvspaf

For an oMPG $\mpg{A}$ that has no exits, there is no pending entrance, and an entrance is winning if and only if the corresponding initial position of the equivalent MPG is winning (in the conventional sense). Indeed, the entrance $i$ of $\mpg{A}$ is winning if there is an $\eve$-strategy $\tau_{\eve}$ such that for any $\adam$-strategy $\tau_{\adam}$, $\denoplay{\indplay{\tau_{\eve}}{\tau_{\adam}}{i}}{\mpg{A}} = \const{\eve}$, and losing otherwise.

\begin{example}
\myvspin
Let $\mpg{A}:(3, 1)\rightarrow (2, 1)$ be the oMPG in~\cref{fig:openMPG}, where the shape of each position indicates the role, circles for player $\eve$ and rectangles for player $\adam$, and the label corresponds to the assigned weight. 
Note that each of $1_{\dl}, 1_{\dr}, 2_{\dr}$ labels an entrance and an exit in~\cref{fig:openMPG}; for distinction, 
we write $1'_{\dl}, 1'_{\dr}, 2'_{\dr}$ for the \emph{exits} labeled by $1_{\dl}, 1_{\dr}, 2_{\dr}$. 
Then $\denoentry{1_{\dr}}{\mpg{A}} = \big\{\{ 1'_{\dr}\} \big\}$, and
$\denoentry{2_{\dr}}{\mpg{A}} = \big\{\{ (0.6, 1'_{\dl}), (-2.5, 2'_{\dr})\},\ \{ \const{\eve}, (-2.5, 2'_{\dr})\}\big\}$.
We can similarly calculate $\denoentry{3_{\dr}}{\mpg{A}}$ and $\denoentry{1_{\dl}}{\mpg{A}}$.
All of the four entrances are pending.
\end{example}
\myvspaf

\subsection{Traced Symmetric Monoidal Category of Rightward Open MPGs}
\label{sec:TSMCofROMPG}

We introduce a traced symmetric monoidal category (TSMC) $\roMPG$ of ``unidirectional'' oMPGs.
We call the latter \emph{rightward open MPGs (roMPGs)}, since they are defined as oMPGs whose open ends are limited to rightward ones. 
Later in \cref{sec:compCCofOMPGbyInt} we apply the Int construction  to $\roMPG$ (\cref{fig:openMPG}).

\begin{definition}[rightward open MPG (roMPG)]\label{def:roMPG}
\myvspin
A \emph{rightward open MPG} $\mpg{A} = (\bfn{m}, \bfn{n}, Q, E, \rho, w)$ is an oMPG such that $\bfn{m} = (m, 0)$ and $\bfn{n} = (n, 0)$. 
In this case, we say that an roMPG $\mpg{A}$ is from $m$ to $n$, writing $\mpg{A}:m\rightarrow n$.
\end{definition}
\myvspaf

 \begin{figure}[tbp]
     \centering
     \scalebox{0.8}{
     \begin{tikzpicture}[
              innode/.style={draw, rectangle, minimum size=1cm},
              innodemini/.style={draw, rectangle, minimum size=0.5cm},
              interface/.style={inner sep=0}
              ]
              \node[interface] (rdo1) at (0.5cm, -0.25cm) {$1$};
              \node[innode] (game1) at (1.5cm, 0cm) {$\mpg{A}$};
              \node[innodemini] (game2) at (3cm, 0.25cm) {$\mpg{B}$};
              \node[innode] (game3) at (4.5cm, 0cm) {$\mpg{C}$};
              \node[interface] (rcdo1) at (5.5cm, -0.25cm) {$1$};
              \draw[->] (rdo1) to ($(game1.south west)!0.5!(game1.west)$);
              \draw[->] ($(game1.south east)!0.5!(game1.east)$) to ($(game3.south west)!0.5!(game3.west)$);
              \draw[->] ($(game3.south east)!0.5!(game3.east)$) to (rcdo1);
              \draw[->] ($(game1.north east)!0.5!(game1.east)$) to (game2.west);
              \draw[->] (game2.east) to ($(game3.north west)!0.5!(game3.west)$);
              \draw[->] (1cm, 0.75cm) arc [radius=0.3, start angle = 90, end angle=270];
              \draw[<-] (1cm, 0.75cm) to (5cm, 0.75cm);
              \draw[->] (5cm, 0.15cm) arc [radius=0.3, start angle = 270, end angle=450];
              \node[interface] (equal) at (6cm, 0cm) {$=$}; 
              \node[interface] (trl) at (6.9cm, 0cm) {$\trsyb_{1;1,1}\Biggl($};
              \node[interface] (trr) at (13.8cm, 0cm) {$\Biggr)$};
              \node[interface] (rdo11) at (7.5cm, 0.25cm) {$1$};
              \node[interface] (rdo12) at (7.5cm, -0.25cm) {$2$};
              \node[innode] (game11) at (8.5cm, 0cm) {$\mpg{A}$};
              \node[interface] (seq1) at (9.5cm, 0cm) {\scalebox{1.5}{$\seqcomp$}};
              \node[innodemini] (game22) at (10.5cm, 0.25cm) {$\mpg{B}$};
              \node[interface] (oplus1) at (10.5cm, -0.2cm) {$\oplus$};
              \draw[->] (9.7cm, -0.4cm) to (11.3cm, -0.4cm);
              \node[interface] (seq2) at (11.5cm, 0cm) {\scalebox{1.5}{$\seqcomp$}};
              \node[innode] (game32) at (12.5cm, 0cm) {$\mpg{C}$};
              \node[interface] (rcdo11) at (13.5cm, 0.25cm) {$1$};
              \node[interface] (rcdo12) at (13.5cm, -0.25cm) {$2$};
              \draw[->] (rdo11) to ($(game11.west)!0.5!(game11.north west)$);
              \draw[->] (rdo12) to ($(game11.west)!0.5!(game11.south west)$);
              \draw[->] ($(game11.east)!0.5!(game11.north east)$) to (9.3cm, 0.25cm);
              \draw[->] ($(game11.east)!0.5!(game11.south east)$) to (9.3cm, -0.25cm);
              \draw[->] (9.7cm, 0.25cm) to (game22.west);
              \draw[->] (game22.east) to (11.3cm, 0.25cm);
              \draw[->] (11.7cm, 0.25cm) to ($(game32.west)!0.5!(game32.north west)$);
              \draw[->] (11.7cm, -0.25cm) to ($(game32.west)!0.5!(game32.south west)$);
             \draw[->] ($(game32.east)!0.5!(game32.north east)$) to (rcdo11);
             \draw[->] ($(game32.east)!0.5!(game32.south east)$) to (rcdo12);
          \end{tikzpicture}
          }
     \caption{A string diagram of MPGs, with algebraic operations $\seqcomp$, $\oplus$, $\trsyb$.}
     \label{fig:algebraicOperations}
 \end{figure}
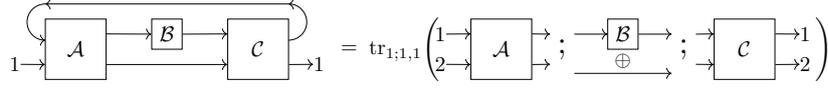

To show that roMPGs form a TSMC, we define sequential composition $\seqcomp$, sum $\oplus$, and trace operator $\trsyb$ on roMPGs.
Although the following definitions of $\seqcomp, \oplus, \trsyb$ look complicated, the intuition behind them is quite clear: see \cref{fig:seqCompOplusIllustrated,fig:trace_operator},
restricting to oMPGs whose open ends are only rightward.
\cref{fig:algebraicOperations} shows an example involving all three operations.

The sequential composition $\mpg{A}\seqcomp\mpg{B}$ of roMPGs $\mpg{A}$ and $\mpg{B}$ connects (and hides) the exits in $\mpg{A}$ and the corresponding entrances in $\mpg{B}$.

\begin{definition}[sequential composition $\seqcomp$ of roMPGs]
\myvspin
    Let $\mpg{A}:m\rightarrow l$ and $\mpg{B}:l\rightarrow n$ be roMPGs.
    Their sequential composition $\mpg{A}\seqcomp \mpg{B}$ is given by $\mpg{A}\seqcomp \mpg{B}\defeq (m, n, Q^{\mpg{A}} + Q^{\mpg{B}},E^{\mpg{A}\seqcomp \mpg{B}}, [\rho^{\mpg{A}},\rho^{\mpg{B}}], [w^{\mpg{A}},w^{\mpg{B}}])$, where
(i)  $E^{\mpg{A}\seqcomp \mpg{B}}$ is defined in the following natural manner:
\myvspmathbef
\begin{myminipage}
  \begin{alignat*}{3}
      & \text{for $s\in \nset{m}+Q^{\mpg{A}}$ and $s'\in Q^{\mpg{A}} $, }
      && \text{$(s, s')\in E^{\mpg{A}\seqcomp \mpg{B}} \text{ if } (s, s')\in E^{\mpg{A}}$,}
    \\& \text{for $s\in \nset{m} + Q^{\mpg{A}}$ and $s'\in \nset{n}+ Q^{\mpg{B}} $, }
      && \text{$(s, s')\in E^{\mpg{A}\seqcomp \mpg{B}} \text{ if } \exists i\in \nset{l}. \ (s, i)\in  E^{\mpg{A}} \land (i, s')\in E^{\mpg{B}}$,}
    \\& \text{for $s\in Q^{\mpg{B}}$ and $s'\in \nset{n}+ Q^{\mpg{B}} $, }
      && \text{$(s, s')\in E^{\mpg{A}\seqcomp \mpg{B}} \text{ if } (s, s')\in E^{\mpg{B}}$,}
    \\& \text{for $s\in Q^{\mpg{B}}$ and $s'\in Q^{\mpg{A}} $,}  
      && \text{$(s, s') \notin E^{\mpg{A}\seqcomp \mpg{B}}$,}
  \end{alignat*}
\end{myminipage}
\myvspmathaf
      (ii) $[\rho^{\mpg{A}},\rho^{\mpg{B}}]: Q^{\mpg{A}} + Q^{\mpg{B}} \rightarrow \{\eve, \adam\}$ combines $\rho^{\mpg{A}}$, $\rho^{\mpg{B}}$ by case distinction, and similarly
      (iii) $[w^{\mpg{A}},w^{\mpg{B}}]: Q^{\mpg{A}} + Q^{\mpg{B}} \rightarrow \Real$
      combines $w^{\mpg{A}}$, $w^{\mpg{B}}$ by case distinction.
\end{definition}
\myvspaf

Defining sum $\oplus$ of roMPGs is straightforward; see \cref{sec:DefOfROMPG}. 

The trace operator $\trsyb$ makes loops in roMPGs, from ``upper'' side (see \cref{fig:trace_operator}).
This plays an important role for the ``bidirectional'' sequential composition in the Int construction.

 \begin{definition}[trace operator $\trace{l}{m}{n}{\mpg{A}}$ over roMPGs]\label{def:trROMPG}
\myvspin
Let $\mpg{A}:l+m\rightarrow l+n$ be an roMPG. The trace $\trace{l}{m}{n}{\mpg{A}}:m\rightarrow n$ of $\mpg{A}$ is given by $\trace{l}{m}{n}{\mpg{A}}\defeq \big( m, n, Q^{\mpg{A}}, E^{\trace{l}{m}{n}{\mpg{A}}}, \rho^{\mpg{A}}, w^{\mpg{A}}\big)$, where 
$E^{\trace{l}{m}{n}{\mpg{A}}} \defeq 
\big\{ (s, s')\in (\nset{m} + Q^{\mpg{A}})\times (\nset{n}+Q^{\mpg{A}}) \ \big| \ 
\exists k\in \Nat.\ \forall j \in \nset{k}.\ \forall i_j\in \nset{l}.\allowbreak \ 
\lplus{l}{s} \rightarrow_{\mpg{A}} i_1 \rightarrow_{\mpg{A}} \cdots \rightarrow_{\mpg{A}} i_k \rightarrow_{\mpg{A}} \lplus{l}{s'} 
\big\}$, 
and for $l\in \Nat$ and $s\in \nset{m}+Q$ we define $\lplus{l}{s}\in \nset{l+m}+Q$ by: $\lplus{l}{s} \defeq l + s$ if $s\in \nset{m}$, and $\lplus{l}{s} \defeq s$ if $s\in Q$.
 \end{definition}
\myvspaf

Note that $i_1,\dotsc, i_k$ above are open ends (and not positions).

The algebraic operations $\seqcomp, \oplus, \trsyb$ on roMPGs (with trivial constants, identity $\id$ and swap $\swapsyb$) satisfy the \emph{equational axioms of traced symmetric monoidal category (TSMC)},
whose definition is omitted; see~\cite{Watanabe21,Watanabe23} for details.
Precisely, 
for this
we need to define the \emph{roMPG isomorphisms} (see~\cref{sec:DefOfROMPG}) and consider the quotient of the equivalence relation, where two games are equivalent if there is an isomorphism between them.

Finally, the \emph{unidirectional syntactic category} $\roMPG$ is defined in the following statement. 

\begin{proposition}[TSMC $\roMPG$]
\myvspin
Let $\roMPG$ be the category whose objects are natural numbers, whose arrows are
(equivalence classes of) roMPGs,
and whose identity and composition are $\id$ and $\seqcomp$. Then the data $(\roMPG, \oplus, 0, \trsyb)$ constitutes a (strict) TSMC.
\qed
\end{proposition}
\myvspaf

\subsection{Compact Closed Category of Open MPGs}
\label{sec:compCCofOMPGbyInt}

Finally, we define the \emph{bidirectional} syntactic category $\oMPG$ for oMPGs, by applying to $\roMPG$ the \emph{$\Int$ construction}~\cite{joyal1996}, a general construction from TSMCs to compCCs. Thus, in terms of category theory, the following definition means that $\oMPG \defeq \Int(\roMPG)$.

\begin{definition}[category $\oMPG$]
\myvspin\label{def:oMPGCat}
 The \emph{category $\oMPG$ of open MPGs} is defined as follows. Its objects are pairs $(\nfr{m},\nfl{m})$ of natural numbers. Its arrows are defined by rightward open MPGs as follows,
 where the double line $=\joinrel=$ means ``is the same thing as'':
\myvspmathbef[-0.8em]
\begin{myminipage}
\begin{equation}\label{eq:arrowsOfOMPGCat}
 \vcenter{\infer={ 
   \text{an arrow } \mpg{A}\colon \nfr{m}+\nfl{n}\longrightarrow  \nfr{n}+\nfl{m}
   \text{ in $\roMPG$, i.e.\ an roMPG}
   }{
   \text{an arrow } (\nfr{m}, \nfl{m}) \longrightarrow (\nfr{n}, \nfl{n}) 
   \text{ in $\oMPG$}
   }}
\end{equation}
\end{myminipage}
\myvspmathaf
\end{definition}
\myvspaf[-1.2em]

For an oMPG $((\nfr{m}, \nfl{m}), (\nfr{n}, \nfl{n}), Q, E,\rho, w)$ (in the style of \cref{def:oMPG}),
the corresponding arrow from $(\nfr{m}, \nfl{m})$ to $(\nfr{n}, \nfl{n})$ in $\oMPG$---which by \cref{def:oMPGCat}  must be an roMPG---is $((\nfr{m}+\nfl{n},0), (\nfr{n}+\nfl{m}, 0), Q, E,\rho, w)$. 
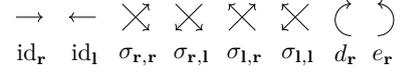
\begin{wrapfigure}[4]{r}{0pt}
\scalebox{0.7}{
  \begin{minipage}[t]{20em}
  \vspace{-1.5em}
    \begin{tikzpicture}[
              innode/.style={draw, rectangle, minimum size=1cm},
              innodemini/.style={draw, rectangle, minimum size=0.5cm},
              interface/.style={inner sep=0},
              innodeeve/.style={draw, circle, minimum size=0.2cm},
              innodeadam/.style={draw, rectangle, minimum size=0.2cm},
              ]
              \draw[->] (0cm, 0cm) to (0.5cm, 0cm);
              \node[interface] (idr) at (0.3cm, -0.7cm) {\scalebox{1.3}{$\id_{\dr}$}};
              \draw[<-] (1cm, 0cm) to (1.5cm, 0cm);
              \node[interface] (idl) at (1.3cm, -0.7cm) {\scalebox{1.3}{$\id_{\dl}$}};
              \draw[->] (2cm, 0.25cm) to (2.5cm, -0.25cm);
              \draw[->] (2cm, -0.25cm) to (2.5cm, 0.25cm);
              \node[interface] (swaprr) at (2.3cm, -0.75cm) {\scalebox{1.3}{$\swap{\dr}{\dr}$}};
              \draw[->] (3cm, 0.25cm) to (3.5cm, -0.25cm);
              \draw[<-] (3cm, -0.25cm) to (3.5cm, 0.25cm);
              \node[interface] (swaprl) at (3.3cm, -0.75cm) {\scalebox{1.3}{$\swap{\dr}{\dl}$}};
              \draw[<-] (4cm, 0.25cm) to (4.5cm, -0.25cm);
              \draw[->] (4cm, -0.25cm) to (4.5cm, 0.25cm);
              \node[interface] (swaplr) at (4.3cm, -0.75cm) {\scalebox{1.3}{$\swap{\dl}{\dr}$}};
              \draw[<-] (5cm, 0.25cm) to (5.5cm, -0.25cm);
              \draw[<-] (5cm, -0.25cm) to (5.5cm, 0.25cm);
              \node[interface] (swapll) at (5.3cm, -0.75cm) {\scalebox{1.3}{$\swap{\dl}{\dl}$}};
              \draw[<-] (6.3cm, 0.25cm) arc [radius=0.3, start angle = 90, end angle=270];
              \node[interface] (dr) at (6.2cm, -0.7cm) {\scalebox{1.3}{$d_{\dr}$}};
              \draw[->] (6.8cm, -0.35cm) arc [radius=0.3, start angle = 270, end angle=450];
              \node[interface] (er) at (6.9cm, -0.75cm) {\scalebox{1.3}{$e_{\dr}$}};
          \end{tikzpicture}
    \vspace{-2em}
    \end{minipage}
    }
    \caption{Constants of compCC.}
    \label{fig:id-swap-co-unit}
  \end{wrapfigure}
  
The compact closed category $\oMPG$ has two algebraic operations $\seqcomp$ and $\oplus$, as illustrated 
 in \cref{fig:seqCompOplusIllustrated}.
The rest of the compact closed structure consists of  the constants shown in \cref{fig:id-swap-co-unit}, namely
 \emph{identity} ($\id$), \emph{swap} ($\swapsyb$), \emph{unit} ($d$), and \emph{counit} ($e$).
%
All these algebraic operations are automatically derived by the $\Int$ construction.
The explicit definitions of $\seqcomp$ and $\oplus$ can be found in \cref{sec:DefOfOMPG}.

The Int construction also ensures that
the algebraic operations of $\oMPG$ above satisfy \emph{the equational axioms of compCCs} 
(see~\cite{joyal1996}).

\begin{theorem}[$\oMPG$ is a compCC]
\myvspin\label{thm:oMDPIsCompactClosed}
 The category $\oMPG$ 
 is a compact closed category.
\qed
\end{theorem}
\myvspaf

\section{Decomposition Equalities for Memoryless Plays}
\label{sec:decEq}

Here we introduce \emph{decomposition equalities}, a cornerstone of our compositional  MPG solution. 

Looking at our workflow (\cref{fig:catsFunctors}), 
the bottom level $\fpwpFunctor\colon \roPlay\to\fpsemCat$ is about ``plays'': 
the category $\roPlay$ has as its arrows \emph{rightward open play graphs (roPGs)}---they are roughly 
 (indexed families of) memoryless plays.
Once we get $\fpwpFunctor$ compositional (i.e.\ preserving traced monoidal structures), this compositionality $\fpwpFunctor$  carries over to the middle and top levels of \cref{fig:catsFunctors} via categorical constructions, eventually realizing a compositional MPG solution $\fwpFunctor$. Decomposition equalities are a key to the compositionality of the bottom level $\fpwpFunctor$.




We first define an roPG. It is intuitively a collection of  plays, one for each entrance $i$, indexed by $i$.

\begin{definition}[roPG]
\myvspin
A \emph{rightward open play graph (roPG)} from $m$ to $n$ is an roMPG from $m$ to $n$ (\cref{def:roMPG})  whose set of edges $E$ is a partial function (i.e.\ at most one successor).
\end{definition}
\myvspaf

Rightward open play graphs forms a TSMC $\roPlay$ just in the same way as $\roMPG$---$\roPlay$ is a subcategory of $\roMPG$. Due to its determinancy, 
an roPG $\mpg{C}:m\rightarrow n$ has a unique play from each entrance $i \in \nset{m}$; this play is denoted by $\indplaya{\mpg{C}}{i}$. Note that $\indplaya{\mpg{C}}{i}$ is also the play induced by the unique $\eve$- and $\adam$-strategies of $\mpg{C}$, and thus is memoryless by definition.



 We start with the decomposition equality for $\seqcomp$.

\begin{proposition}[decomposition equality for $\seqcomp$]
\myvspin
\label{prop:deqSeqc}
    Let $\mpg{C}:m\rightarrow l$, $\mpg{D}:l\rightarrow n$ be roPGs, and $i\in \nset{m}$.
    The following equality holds regarding the denotation $\denoplay{\_\,}{}$ of plays (\cref{def:denoPlays}):
    \myvspmathbef
    \begin{myminipage}
    \begin{align*}
    \denoplay{\indplaya{\mpg{C}\seqcomp \mpg{D}}{i}}{\mpg{C}\seqcomp \mpg{D}} & =
        \begin{cases}
             \denoplay{\indplaya{\mpg{C}}{i}}{\mpg{C}} 
           & \text{ if } \denoplay{\indplaya{\mpg{C}}{i}}{\mpg{C}} \in\{\const{\eve}, \const{\adam}\},\\
             \denoplay{\indplaya{\mpg{D}}{j}}{\mpg{D}} 
           & \text{ if } \denoplay{\indplaya{\mpg{C}}{i}}{\mpg{C}}=j \in \nset{l}, \\
             \denoplay{\indplaya{\mpg{D}}{j}}{\mpg{D}} 
           & \text{ if } \denoplay{\indplaya{\mpg{C}}{i}}{\mpg{C}}=(r, j) \in \Real \times \nset{l}\text{ and }\denoplay{\indplaya{\mpg{D}}{j}}{\mpg{D}}\in \{\const{\eve}, \const{\adam}\}, \\
             (r, k) 
           &\text{ if } \denoplay{\indplaya{\mpg{C}}{i}}{\mpg{C}}=(r, j)\in \Real \times \nset{l} \text{ and }\denoplay{\indplaya{\mpg{D}}{j}}{\mpg{D}} = k\in \nset{n},\\
             (r+r', k) 
           &\text{ if } \denoplay{\indplaya{\mpg{C}}{i}}{\mpg{C}}=(r, j)\in \Real \times \nset{l} \text{ and }\denoplay{\indplaya{\mpg{D}}{j}}{\mpg{D}} = (r', k)\in \Real \times \nset{n}.
          \qed
        \end{cases}
    \end{align*}
    \end{myminipage}
   \myvspmathaf
\end{proposition}
\myvspaf[-1em]

Here is some intuition.
In the first case, the winner of the composed play $\indplaya{\mpg{C}\seqcomp \mpg{D}}{i}$ is already decided within $\mpg{C}$.
In the second case, if the play $\indplaya{\mpg{C}\seqcomp \mpg{D}}{i}$ goes immediately to exit $j$, then the denotation of $\indplaya{\mpg{C}\seqcomp \mpg{D}}{i}$  solely relies on $\mpg{D}$. The third case models \emph{prefix independence}  of MPGs if an ultimate winner is decided in $\mpg{D}$, then the ``prefix'' in  $\mpg{C}$ does not matter. The fourth case is easy; finally, in the fifth case, we accumulate weights $r,r'$ from $\mpg{C}, \mpg{D}$, respectively.


The decomposition equality for the sum $\oplus$ is easy and is omitted.

We move on to the decomposition equality for the trace operator (cf.\ \cref{def:trROMPG}). We prepare some definitions.
%
Firstly, we define the \emph{traced denotation of plays (TDPs)}
of roPG $\mpg{E}:l+m\rightarrow l+n$.

\begin{definition}[TDP $\tdip{l}{m}{n}{E}{i}$]
\myvspin
\label{def:traced_induced_plays}
Let $\mpg{E}:l+m\rightarrow l+n$ be an roPG, and $i\in\nset{m}$ be an entrance. 
The \emph{traced denotation of  plays (TDP)} 
$\tdip{l}{m}{n}{E}{i}$
of $\mpg{E}$ from $i$
is the (unique) possibly infinite sequence $v= (v_0,v_1,\cdots)$ 
of elements in $\nset{l}+\nset{m}+\nset{n} + \Real\times (\nset{l} + \nset{n}) + \{\const{\eve}, \const{\adam}\}$
that satisfies $v_0 = i$ and the following conditions:
\begin{enumerate}
    \item 
    $v_{j+1} = \denoplay{\indplaya{\mpg{E}}{v_j}}{}$ if $v_j\in \nset{l}+\nset{m}$,
    \item
    $v_{j+1} = \denoplay{\indplaya{\mpg{E}}{k}}{}$ if $v_j = (r, k)\in \Real \times \nset{l}$,
    \item
     $v_{j+1}$ is undefined if $v_j \in \nset{n} + \Real \times \nset{n} + \{\const{\eve}, \const{\adam}\}$.
\end{enumerate}
\end{definition}
\myvspaf


\todo{Get the numbering uniform in the explanation and in the definition}
The last definition can be thought of as a summary of the unique play of $\trace{l}{m}{n}{\mpg{E}}$, where we only record reaching open ends and the winner decided. Specifically, we record 1) reaching $k\in\nset{n}$ (then the play is over), 2) reaching $k\in \nset{l}$ (then the play loops), 3) the winner decided within $\mpg{E}$ (this is when $v_{j}\in \{\const{\eve}, \const{\adam}\}$), and 4) initialization (this is when $v_{j}\in \nset{m}$; this can only happen for $j=0$). Additionally, we record the accumulated weight in its course (note the $\Real\times \_\,$ components in the definition).


We define the \emph{mean payoff (MP) condition} for an infinite TDP $\tdip{l}{m}{n}{E}{i} = v = (v_j)_{j \in \Nat}$ as follows.
Since $v$ is infinite, by Item 3 in \cref{def:traced_induced_plays},
$v_j \in \nset{l} + \Real\times \nset{l}$ for any $j \ge 1$.
We call TDP $v$ \emph{productive} if $v_j \in \Real\times \nset{l}$ for infinitely many $j$. 
In fact, every TDP is productive, which can be easily shown by the condition $(\ddagger)$ in~\cref{def:oMPG} and the pigeonhole principle.
Now let $v' \in \Real^\Nat$ be the infinite sequence obtained by extracting all weights from $v$, where $v_j \in \nset{l}$ (without a weight) is simply skipped.
We say $v$ satisfies \emph{the MP condition} if  $v'$ does.

We also define the denotation of TDP, similarly to~\cref{def:denoPlays}. 
\begin{definition}[denotation of TDP]
\label{def:denotOfTDP}
\myvspin
\label{def:denTracedInducedPlays}
    For roPG $\mpg{E}:l+m\rightarrow l+n$, $i\in \nset{m}$, and the TDP $v\allowbreak=\allowbreak\tdip{l}{m}{n}{E}{i}$ of $\mpg{E}$ from $i$, the \emph{denotation} $\denosdigest{v}\in \nset{n} + \Real\mspace{2mu}{\times}\mspace{2mu}\nset{n} + \{\const{\eve}, \const{\adam}\}$ of $v$
    is defined as: 
    \myvspmathbef
    \begin{myminipage}
    \begin{alignat*}{2}
     & \const{\eve} 
     && \text{if $v$ is infinite and satisfies the MP condition,}\\
     & \const{\adam} 
     && \text{if $v$ is infinite and does not satisfy the MP condition,}\\
     & v_k 
     && \text{if $v = (v_0,\cdots, v_k)$ and $v_k\in \{\const{\eve}, \const{\adam}\}$,}\\
     & v_k 
     && \text{if $v = (v_0,\cdots, v_k)$, $v_k\in \nset{n}$, and $v_j\in \nset{l}$ for each $j\in \nset{k-1}$,}\\
     & \big(\textstyle\sum^{}_{j\in \nset{k}} \mathrm{wt}(v_j),\, v_k\big)
     &&\text{if $v = (v_0,\cdots, v_k)$, $v_k\in \nset{n}$, and $v_j\in \Real\times \nset{l}$ for some $j\in \nset{k-1}$,}\\
     & \big(\textstyle\sum^{}_{j\in \nset{k}} \mathrm{wt}(v_j),\, \pi_2(v_k)\big)
     \ 
     &&\text{if $v = (v_0,\cdots, v_k)$ and $v_k\in \Real\times\nset{n}$,}
    \end{alignat*}
    \end{myminipage}
    \myvspmathaf
    where the \emph{weight-sum} $\sum^{}_{j\in \nset{k}} \mathrm{wt}(v_j)$ is defined as
    $\sum_{j\in \nset{k} \text{ \upshape such that } v_j\in \Real \times (\nset{l} + \nset{n})}\pi_1(v_j)$,
    and $\pi_1:\Real \times (\nset{l} + \nset{n}) \rightarrow \Real$ and $\pi_2:\Real\times (\nset{l} + \nset{n}) \rightarrow \nset{l} + \nset{n}$ are the first and second projections.
\end{definition}
\myvspaf

\omitatsubm{
\begin{remark}
\myvspin
\label{rem:uniqueSuccAccCondForDenOfTDP}
As explained above, the condition $(\ddagger)$ ensures the productivity of infinite TDP $v$.
Without this condition, there would be infinite sequences that have only finitely many $v_j$'s which are assigned a weight; then, it would be unclear how to define the denotation $\denosdigest{v}$.
Similar situations would arise in the definitions of trace operators on the syntactic/semantic categories.
This is why we introduce the condition $(\ddagger)$ (and its semantic counterpart, the realizability condition, given later).
\end{remark}
\myvspaf
}

Finally, we show the \emph{decomposition equality for the trace operator $\trsyb$}. 
The equality says that the behavior of
 $\trace{l}{m}{n}{\mpg{E}}$
can be described by the behavior of $\mpg{E}$---note that $\tdip{l}{m}{n}{E}{i}$ on the right-hand side is described by the denotations of suitable plays of $\mpg{E}$ (\cref{def:traced_induced_plays}).


\begin{proposition}[decomposition equality for $\trsyb$]
\myvspin\label{prop:deqtrace}
    Let $\mpg{E}:l+m\rightarrow l + n$ be an roPG, and $i\in \nset{m}$ be an entrance on $\trace{l}{m}{n}{\mpg{E}}$. The following equality holds:
    \myvspmathbef
    \begin{myminipage}
    \begin{align*}
        \denoplay{\indplaya{\trace{l}{m}{n}{\mpg{E}}}{i}}{} & = 
        \denosdigest{\,\tdip{l}{m}{n}{E}{i}\,}
        .\qed
    \end{align*}
    \end{myminipage}
    \myvspmathaf
\end{proposition}
\myvspaf[-1.3em]

It is still  nontrivial whether the right-hand side of \cref{prop:deqtrace} can be effectively computed---\cref{def:traced_induced_plays} utilizes a possibly infinite sequence. We can exploit the ultimate periodicity that arises from the finiteness of an roPG $\mpg{E}$; this is much like in~\cref{lem:periodicity}. Our implementation (\cref{sec:impAndExp}) uses this technique.

\section{Semantic Categories and Winning-position Functors}
\label{sec:semanticCategories}

Here we give our compositional solution for oMPGs, by defining the semantic category $\fsemCat$ for oMPGs and the winning-position functor $\fwpFunctor : \oMPG \to \fsemCat$, where $\fsemCat$ has and $\fwpFunctor$ preserves the compact closed structure. As shown in \cref{fig:catsFunctors}, we construct $\fsemCat$ and $\fwpFunctor$ in  two steps, via  the change-of-base and Int constructions.

\subsection{Semantic Category and Functor for Plays}
Firstly, we define the semantic category $\fpsemCat$ of roPGs (the bottom level of \cref{fig:catsFunctors}). The development in~\cref{sec:decEq} is crucial here.
 The operation
$\fpmonad(X) \defeq X + \Real\times X + \{\const{\eve},\const{\adam}\}$
used below was introduced previously in \cref{def:denoPlays} for defining the denotation $\denoplay{\pi}{\mpg{A}}$ of a play.

\begin{definition}[objects and arrows of $\fpsemCat$]
\myvspin\label{def:fpsemCat}
The category $\fpsemCat$  is defined as follows. Its objects are natural numbers.
Its arrows $f$ from $m$ to $n$ (denoted by $f : m\rightarrow n$ in $\fpsemCat$) are functions of the type
\begin{math}
 f\colon\nset{m} \longrightarrow\fpmonad(\nset{n})= \nset{n} + \Real\times \nset{n} + \{\const{\eve},\const{\adam}\} 
\end{math}. Such functions are further subject to the \emph{realizability condition}:
if some $i\in\nset{m}$ goes straight to an exit $k\in \nset{n}$, then there should be no other $j\in\nset{m}$ with $j\neq i$ that goes to the same exit $k$, with or without weights. To put it precisely: if $ f(i) =k\in \nset{n}$, then for each $j\in\nset{m}\setminus \{i\}$,  $f(j)\neq k$ and $f(j)\neq (r,k)$ for any $r$.

\end{definition}
\myvspaf


The realizability condition corresponds to the condition $(\ddagger)$ in~\cref{def:oMPG}.

We move on to the definition of the algebraic operations of the traced symmetric monoidal category (TSMC) $\fpsemCat$, i.e., $\,\seqcomp\,$, $\oplus$, and $\trsyb$. The sequential composition $\seqcomp$ and the trace operator $\trsyb$ are defined in the same way as~\cref{prop:deqSeqc,prop:deqtrace}, and the definition of the sum $\oplus$ is clear.

\begin{definition}[sequential composition $\seqcomp$ of $\fpsemCat$]
\myvspin
    Let $f:m\rightarrow l$ and $g:l\rightarrow n $ be arrows in $\fpsemCat$, and $i\in \nset{m}$. Their \emph{sequential composition} $f \mspace{2mu}{\seqcomp}\mspace{2mu} g : m\rightarrow n$ of $f$ and $g$ is given as follows:
    \myvspmathbef
    \begin{myminipage}
    \begin{align*}
       \big(f\seqcomp g\big)(i)  &=  
        \begin{cases}
            f(i) 
           & \text{ if } f(i) \in\{\const{\eve}, \const{\adam}\},\\
            g(j) 
           & \text{ if } f(i)=j \in \nset{l}, \\
            g(j) 
           & \text{ if } f(i) =(r, j) \in \Real \times \nset{l},\text{ and }g(j)\in \{\const{\eve}, \const{\adam}\}, \\
            (r, k) 
           &\text{ if } f(i)=(r, j)\in \Real \times \nset{l} \text{, and }g(j) = k\in \nset{n},\\
            (r+r', k) 
           &\text{ if } f(i)=(r, j)\in \Real \times \nset{l}, \text{ and }g(j) = (r', k)\in \Real \times \nset{n}.
        \end{cases}
    \end{align*}
    \end{myminipage}
    \myvspmathaf
\end{definition}
\myvspaf[-1.2em]

The trace operator $\trsyb$ can also be defined in the same manner (see \cref{sec:proofFplayTSMC}). 

Behind the definitions of the symmetric monoidal category $(\fpsemCat, \seqcomp, \oplus)$,
we can find and utilize a categorical concept of \emph{monad}, which models the notion of computation~\cite{Moggi91}.
Specifically, the mapping of a set $X$ to the set $\fpmonad(X) = X + \Real\times X + \{\const{\eve} + \const{\adam}\}$ (used in \cref{{def:fpsemCat}}) extends to a monad on the category $\Sets$ of sets and functions, and then $\fpsemCat$ is a subcategory of the \emph{Kleisli category} of $T$, where $\oplus$ is defined as the coproduct (see~\cref{sec:play_monads} for the details).

What remains to be shown for the next proposition is that the trace operator satisfies the axioms of trace operator~\cite{joyal1996}.
The proof is lengthy, but straightforward once we find that we can use a similar technique to \cref{lem:periodicity}; see~\cref{sec:proofFplayTSMC} for the proof.
\todo{In an extended version, add description of what kind of finiteness (namely that of interface, not that of a position space) is crucial for what (namely dinaturality).}
\begin{proposition}[TSMC $\fpsemCat$]
\myvspin
\label{prop:fplayTSMC}
The category $\fpsemCat$ is a TSMC.
\qed
\end{proposition}
\myvspaf

The \emph{solution functor} $\fpwpFunctor$ maps roPG $\mpg{C}$ to the denotations $(\denoplay{\indplaya{\mpg{C}}{i}}{})_i$ of its plays:

\begin{definition}[solution functor $\fpwpFunctor$]
\myvspin
    The \emph{solution functor} $\fpwpFunctor:\roPlay\rightarrow \fpsemCat$ is defined as follows: the mapping on objects is given by $\fpwpFunctor(m) \defeq m$, and for an arrow $\mpg{C}\in \roPlay(m, n)$, we define $\fpwpFunctor(\mpg{C}) \in \fpsemCat(m, n)$ as $\nset{m} \ni i \mapsto \denoplay{\indplaya{\mpg{C}}{i}}{} \in T(\nset{n})$. Here
$\denoplay{\indplaya{\mpg{C}}{i}}{}$ is from
\cref{def:denoPlays}.
\end{definition}
\myvspaf

The following is the categorical reformulation of the key results, \cref{prop:deqSeqc,prop:deqtrace}.

\begin{theorem}[compositionality for play graphs]
\myvspin
\label{thm:fpwpfuncComp}
$\fpwpFunctor:\roPlay\rightarrow \fpsemCat$ 
is a traced symmetric monoidal functor.
In particular,
$\, \fpwpFunctor(\mpg{C}\seqcomp \mpg{D}) = \fpwpFunctor(\mpg{C})\seqcomp\fpwpFunctor(\mpg{D})$,
$\, \fpwpFunctor(\mpg{C}\oplus \mpg{D}) = \fpwpFunctor(\mpg{C})\oplus\fpwpFunctor(\mpg{D})$,
and 
$\, \fpwpFunctor(\trsyb(\mpg{E})) = \trsyb(\fpwpFunctor(\mpg{E}))$.
\end{theorem}
\myvspaf

\subsection{Semantic Category and Functor for roMPGs}


We move on to the middle level of \cref{fig:catsFunctors}.
We construct the semantic category $\frsemCat$ 
for roMPGs by the change of base construction~\cite{Eilenberg65,cruttwell2008normed} from $\fpsemCat$ 
for roPGs.
We give the definition concretely below, but in categorical terms, $\frsemCat$ 
is obtained simply by applying to $\fpsemCat$ 
the change of base construction by 
the iterated finite powerset functor $\powerfuncf \circ \powerfuncf$ on $\Sets$.

\begin{definition}[objects and arrows of $\frsemCat$]
\myvspin
    The category $\frsemCat$ has natural numbers $m$ as objects. Its arrow $F:m\rightarrow n$ is 
    an element in $\powerfuncf(\powerfuncf(\fpsemCat(m,n)))$, i.e.,
    a set $\big\{\{ f_{i, j}:m\rightarrow n \text{ in } \fpsemCat \mid i\in I_j\}\allowbreak \ \big|\allowbreak\  j\in J \big\}$ of sets of arrows from $m$ to $n$ in $\fpsemCat$
    ($J$ and $I_j$ are arbitrary finite index sets).
\end{definition}
\myvspaf

The intuition of the above definition is that index $j$ in the outer set represents an $\eve$-strategy, and index $i$ in the inner set represents an $\adam$-strategy. Once an $\eve$-strategy $\tau_{\eve}$ (corresponding to $j$) and a $\adam$-strategy $\tau_{\adam}$ (corresponding to $i$) are fixed, the arrow $f_{i, j}$ corresponds to the denotations of the plays induced by them. 
\begin{definition}[sequential composition $\seqcomp$ of $\frsemCat$]
\myvspin
Let $F:m\rightarrow l$, $G:l\rightarrow n$ be arrows in $\frsemCat$. Their sequential composition $F\seqcomp G$ is given by $F\seqcomp G \defeq \big\{ \{f\seqcomp g\mid f\in F',\ g\in G'\}\ \big|\allowbreak\ F'\in F,\ G'\in G\big\}$, where $f\seqcomp g$ is the sequential composition in $\fpsemCat$.
\end{definition}
\myvspaf

The sum $\oplus$ and trace $\trsyb$ are similarly defined by applying the operations of $\fpsemCat$ elementwise.

\begin{proposition}
\myvspin
    $\frsemCat$ is a TSMC.
    \qed
\end{proposition}
\myvspaf

Next we define the semantic functor $\frwpFunctor$ for roMPGs, using $\fpwpFunctor$ for roPGs.
To connect the notion of roMPG to that of roPG,
we define \emph{induced roPG}, similarly to memoryless play.


\begin{definition}[induced roPG $\sybpg{\mpg{A}}{\tau^{\eve}}{\tau^{\adam}}$]
\myvspin
\label{def:roPGfromMPGandStrategies}
Let $\mpg{A}:m\rightarrow n$ be an roMPG, and $\tau^{\eve}$ and $\tau^{\adam}$ be (memoryless) $\eve$- and $\adam$-strategies on $\mpg{A}$, respectively. 
The \emph{induced roPG} $\sybpg{\mpg{A}}{\tau^{\eve}}{\tau^{\adam}}$ on $\mpg{A}:m\rightarrow n$ by $\tau^{\eve}$ and $\tau^{\adam}$
is defined as $(m, n, Q^{\mpg{A}}, E, \rho^{\mpg{A}}, w^{\mpg{A}})$ where partial function $E$ is defined as follows.
For $i\in \nset{m}$, $E(i) \defeq E^{\mpg{A}}(i)$. For $s\in Q $, $E(s) \defeq \tau^{\eve}(s)$ if $\tau^{\eve}(s)$ is defined, $E(s) \defeq \tau^{\adam}(s)$ if $\tau^{\adam}(s)$ is defined, and $E(s)$ is undefined otherwise.
\end{definition}
\myvspaf

We can easily check that $\indplaya{\sybpg{\mpg{A}}{\tau^{\eve}}{\tau^{\adam}}}{i}$ is $\indplay{\tau_{\eve}}{\tau_{\adam}}{i}$ as a play on $\mpg{A}$. This is used in:

\begin{definition}[rightward winning-position functor $\frwpFunctor$]
\label{def:rwinposfunctor}
\myvspin
The \emph{rightward winning-position functor} $\frwpFunctor:\roMPG\rightarrow \frsemCat$ is defined as follows.
The mapping on objects is given by $\frwpFunctor(m) \defeq m$.
For an arrow $\mpg{A}\in \roMPG(m, n)$, we define $\frwpFunctor(\mpg{A}) \in \frsemCat(m, n)$ by 
\myvspmathbef[-0.5em]
\begin{myminipage}
\begin{align*}
    \frwpFunctor(\mpg{A})\defeq \Big\{ \big\{  \fpwpFunctor\big(\sybpg{\mpg{A}}{\tau^{\eve}}{\tau^{\adam}}\big) 
\big|\  \tau_{\adam}\in  \stra{\adam}{\mpg{A}}\big\}\ \Big|\ \tau_{\eve}\in \stra{\eve}{\mpg{A}}\Big\};
\end{align*}
note that $\fpwpFunctor\big(\sybpg{\mpg{A}}{\tau^{\eve}}{\tau^{\adam}}\big) $ maps $i$ to
$\denoplay{\indplay{\tau_{\eve}}{\tau_{\adam}}{i}}{\mpg{A}}$.
\end{myminipage}
\end{definition}
\myvspaf

Via the traced symmetric monoidal functor $\fpwpFunctor$ with the change of base technique,
we can establish the compositionality result below.


\begin{theorem}[compositionality for rightward oMPGs]
\myvspin
$\frwpFunctor \colon \roMPG\to \frsemCat$ 
is a traced symmetric monidal functor, preserving $\seqcomp,\oplus,\trsyb$ as in~\cref{thm:fpwpfuncComp}.
\qed
\end{theorem}
\myvspaf

\begin{remark}[Kleisli construction or change of base]
\label{rem:KleisliOrChangeOfBase}
In the compositional approach for \emph{parity games}~\cite[cf.\ Rem.~4.9]{Watanabe21}, the non-deterministic structures of the semantic category and functor are constructed not by the change of base construction but by the Kleisli construction used e.g.\ in~\cite{grellois2015finitary}.
In this Kleisli approach for parity games, an algorithmic result comes from the finitary models~\cite{grellois2015finitary}.
It seems difficult to obtain a finitary model for MPGs based on the Kleisli approach that induces an algorithm since there are infinitely many priorities (while for parity games, there are finitely many).

Even if one obtains some finitary model, another question is whether the \emph{trace operator} is computable. A computationally tractable trace operator will probably only consider memoryless strategies---exploiting memoryless determinacy---which is easy to enforce in the change of base approach  but not easy in the Kleisli approach.

%
\end{remark}

\subsection{Semantic Category and Functor for oMPGs}

Finally, we move on to the top level of \cref{fig:catsFunctors}, and we define the \emph{semantic category $\fsemCat$ for oMPGs}
by the Int constriction. We used the Int constriction already 
in~\cref{def:oMPGCat}.
We simply have $\fsemCat \defeq \Int(\frsemCat)$,
but we give the concrete definition of $\fsemCat$:
\begin{definition}[semantic category $\fsemCat$]
\myvspin
\label{def:mor_in_s}
We define the category $\fsemCat$ as follows.  Its objects are pairs $(\nfr{m},\nfl{m})$ of natural numbers. Its arrows are given by arrows in $\frsemCat$ as follows:
\myvspmathbef[-0.5em]
\begin{myminipage}
\begin{equation*} 
 \vcenter{\infer={ 
   \text{an arrow } F\colon \nfr{m}+\nfl{n}\longrightarrow  \nfr{n}+\nfl{m}
   \text{ in $\frsemCat$}
   }{
   \text{an arrow } F\colon (\nfr{m}, \nfl{m}) \longrightarrow (\nfr{n}, \nfl{n}) 
   \text{ in $\fsemCat$}
   }}
\end{equation*}
\end{myminipage}
\myvspaf[0.1em]

The Int construction ensures that $\fsemCat$ is a compact closed category (compCC).
\end{definition}
\myvspaf

We also obtain the \emph{winning-position} functor $\fwpFunctor$ from the rightward winning-position functor $\frwpFunctor$ by the $\Int$ construction, namely $\fwpFunctor\defeq \Int(\frwpFunctor)$.
Concretely:
\begin{definition}[winning-position functor $\fwpFunctor$]
\myvspin
The \emph{winning-position functor} $\fwpFunctor \colon \oMPG\to \fsemCat$ is defined as follows.
For $\mpg{A}:(m_r, m_l)\rightarrow (n_r, n_l)$, 
\myvspmathbef[-0.5em]
\begin{myminipage}
$$ \fwpFunctor(\mpg{A}) \defeq 
\Big\{ \big\{ 
(i \mapsto \denoplay{\indplay{\tau_{\eve}}{\tau_{\adam}}{i}}{\mpg{A}})
\in \fpsemCat(\nfr{m}+\nfl{n},\ \nfr{n}+\nfl{m}) \ \big|\  \tau_{\adam}\in \stra{\adam}{\mpg{A}}  \big\} \ \Big| \  \tau_{\eve}\in \stra{\eve}{\mpg{A}}\Big\} .
$$
\end{myminipage}
\end{definition}
\myvspaf

We note that the ``compositional'' denotation $\fwpFunctor(\mpg{A})$ above
naturally induces, for each $i \in \nset{\nfr{m}+\nfl{n}}$, the set
$\Big\{ \big\{ 
(\denoplay{\indplay{\tau_{\eve}}{\tau_{\adam}}{i}}{\mpg{A}})
\in T(\nset{\nfr{n}+\nfl{m}}) \ \big|\  \tau_{\adam}\in \stra{\adam}{\mpg{A}}  \big\} \ \Big| \  \tau_{\eve}\in \stra{\eve}{\mpg{A}}\Big\}$, which agrees with the 
``conventional'' denotation $\denoentry{i}{\mpg{A}}$ given in \cref{def:denOfEntrances}.

The following is our main theorem. It is automatically proved by the Int construction. 
\begin{theorem}[compositionality for oMPGs]
\myvspin
\label{thm:fwpFunctorIsCmpactClosed} 
The winning-position functor $\fwpFunctor \colon \oMPG\to \fsemCat$ is a compact closed functor.
That is, $\fwpFunctor$ preserves the operations $\seqcomp$ and $\oplus$ as in
\myvspmathbef[-0.7em]
\begin{myminipage}
\begin{equation*}
  \fwpFunctor(\mpg{A} \seqcomp \mpg{B})
 =  
 \fwpFunctor(\mpg{A}) \seqcomp 
 \fwpFunctor(\mpg{B})
 ,\quad
 \fwpFunctor(\mpg{A}\oplus\mpg{B})
 =  
 \fwpFunctor(\mpg{A})\oplus
 \fwpFunctor(\mpg{B})
\end{equation*}
\end{myminipage}
\myvspmathaf
as well as the constants (see \cref{fig:id-swap-co-unit}).
 \qed
\end{theorem}
\myvspaf

\section{Implementation and Experiment}

\label{sec:impAndExp}
We describe our implementation  $\compMPG$  of the compositional algorithm for oMPGs, and show experimental results. The experiment results 1)
show that our compositional framework has advantages over the state-of-the-art solver QDPM~\cite{Benerecetti20}, and 2)
 identify two major factors that affect the performance of $\compMPG$.
QDPM is a pseudopolynomial algorithm based on
\emph{small progress measure}~\cite{jurdzinski2000small} and \emph{quasi
dominion}~\cite{benerecetti2018solving}. 

\myparagraph{Meager Semantics}
In our implementation of $\compMPG$, to enhance performance, we use \emph{meager semantics},
a refinement of the semantics given in \cref{sec:semanticCategories} (we call it \emph{fat semantics}). 
An arrow $F:\bfn{m}\rightarrow \bfn{n}$ in the fat semantic category $\fsemCat$ collects all possible strategies and some can be redundant; that is, for deciding whether the entrance is winning, losing, or pending, some strategies are dominated by others, and thus can be forgotten.
Since oMPGs have several exits in general, the situation is much like that of multi-objective optimization.
Such dominance relationship between strategies can be described
by a certain order $\leq$.

Specifically, we define the order $\leq$ between arrows in $\fpsemCat$ (the bottom level in \cref{fig:catsFunctors} of plays) as follows:
for arrows $f,g:m\rightarrow n$ in $\fpsemCat$ (which are functions $\nset{m} \to \nset{n} + \Real\times \nset{n} + \{\const{\eve},\const{\adam}\}$), we define $f \leq g$ if for each $i\in \nset{m}$, one of the following conditions is satisfied: 
(i) $f(i) = g(i) \in \nset{n}$,
(ii) $f(i) = (r_f, j_f)$, $g(i) = (r_g, j_g)$, $r_f \geq r_g$ and $j_f = j_g$,
(iii) $f(i) = \const{\eve}$,
or 
(iv) $g(i) = \const{\adam}$. 
We can also formulate this \emph{meager semantics} as a compact closed category/functor similarly to~\cref{sec:semanticCategories}, drawing the same picture as \cref{fig:catsFunctors}. See~\cref{sec:meagerSemantics}.


\myparagraph{Implementation} Our implementation  $\compMPG$ of  our compositional algorithm is based on the meager semantics in Haskell (it is available at \url{https://github.com/Kazuuuuuki/compMPG}). 
We evaluate $\compMPG$ comparing with QDPM. $\compMPG$ takes an oMPG $\mpg{A}$ as input, which is expressed by (a textual format for) a \emph{string diagram} with sequential composition $\seqcomp$, sum $\oplus$, and constants such as $\id$ (\cref{fig:id-swap-co-unit}). See~\cref{fig:algebraicOperations}.

Formally, those inputs are represented in a \emph{free prop} for oMPGs, which is a slight variant of the free prop for open parity games introduced in~\cite{Watanabe21}.
The input also expresses which component is duplicated (such as $\verb|let t = ... in (t; t; t)|$), by which $\compMPG$ can solve the entire input without solving the repeated component $\verb|t|$ more than once.
Given such an oMPG $\mpg{A}$ as input, $\compMPG$ returns the arrow $\fwpFunctor(\mpg{A})$ as output.
If the entire input is an oMPG from $(1, 0)$ to $(0, 0)$ (recall arities from e.g.~\cref{fig:openMPG}), then the input can be interpreted as an MPG whose initial position is the entrance.
We note that $\compMPG$ only decides the winner at each entrance, while QDPM decides that at all positions.

\myparagraph{Experiment Setting}
We pose the following research questions.
\begin{description}
 \item[RQ1] What characteristics of target MPGs affect the execution time of $\compMPG$?
 \item[RQ2] Can $\compMPG$ efficiently solve a variety of MPGs?
 \item[RQ3] Can $\compMPG$ efficiently solve large MPGs?
\end{description}
\myvsplistaf

For evaluating our framework and answering the research questions, we conducted experiments on an Amazon EC2 t2.xlarge instance, 2.30GHz Intel Xeon E5-2686, 4 virtual CPU cores, 16 GB RAM. We built four benchmark sets (a)--(d) to evaluate the three research questions. Each benchmark set consists of 200--400 MPGs; their weights  are randomly assigned in the range $[-100000,100000]$. 

The benchmark sets  (a) and (b) are designed to measure how the compositional structure affects  $\compMPG$. 
The benchmark set (a) evaluates the effect of the \emph{degree of repetition} (DR) in repeated sequential compositions.
DR indicates the frequency of occurrences of repeated parts in the input MPGs, with higher DR meaning more repetition.
The benchmark set (b) assesses the impact of the arity size of oMPGs in sequential compositions. The set (a) has a fixed compositional structure $t(\mpg{A}_{1},\dotsc,\mpg{A}_{n})$; the sub-oMPGs $\mpg{A}_{1},\dotsc,\mpg{A}_{n}$ are randomly generated and resulting in 400 different MPGs; the same for (b).

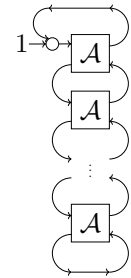
\begin{wrapfigure}[12]{R}{0.12\textwidth}
\centering
    \vspace{-\baselineskip}
  \begin{minipage}[t]{0.11\textwidth}
  \centering
  \scalebox{0.5}{
  \begin{tikzpicture}[
              innode/.style={draw, rectangle, minimum size=1cm},
              interface/.style={inner sep=0},
              innodemini/.style={draw, circle, minimum size=0.2cm}
              ]
              \node[interface] (rdo1) at (-0.3cm, 0.25cm) {\scalebox{2}{$1$}};
              \node[innodemini] (dummy) at (0.5cm, 0.25cm) {};
              \node[innode] (game1) at (1.5cm, 0cm) {\scalebox{2}{$\mpg{A}$}};
              \node[innode] (game2) at (1.5cm, -1.5cm) {\scalebox{2}{$\mpg{A}$}};
              \node[innode] (gamem) at (1.5cm, -4.5cm) {\scalebox{2}{$\mpg{A}$}};
              \draw[->] (rdo1) to (dummy.west);
              \draw[->] (dummy) to ($(game1.north west)!0.5!(game1.west)$);
              \draw[<-] (0.4cm, 1.2cm) to (2cm, 1.2cm);
              \draw[->] (0.4cm, 1.2cm) arc [radius=0.4, start angle = 90, end angle=270];
              \draw[->] (1cm, -0.3cm) arc [radius=0.5, start angle = 90, end angle=270];
              \draw[->] (1cm, -1.8cm) arc [radius=0.5, start angle = 90, end angle=270];
              \draw[->] (1cm, -3.3cm) arc [radius=0.5, start angle = 90, end angle=270];
              \draw[->] (1cm, -4.8cm) arc [radius=0.5, start angle = 90, end angle=270];
              \draw[->] (1cm, -5.8cm) to (2cm, -5.8cm);
              \draw[->] (2cm, -5.8cm) arc [radius=0.5, start angle = 270, end angle=450];
              \draw[->] (2cm, -4.3cm) arc [radius=0.5, start angle = 270, end angle=450];
            \draw[->] (2cm, -2.8cm) arc [radius=0.5, start angle = 270, end angle=450];
            \draw[->] (2cm, -1.3cm) arc [radius=0.5, start angle = 270, end angle=450];
             \draw[->] (2cm, 0.2cm) arc [radius=0.5, start angle = 270, end angle=450];
              \path (game2) -- node[auto=false]{\vdots} (gamem);
          \end{tikzpicture}
        }
    \end{minipage}
    \caption{Benchmark (c): mining.}
    \label{fig:exbenchmark}
\end{wrapfigure}

The benchmark sets (c) and (d) are built for comparison with QDPM.
Each benchmark set contains $200$ randomly generated MPGs that have some compositional structures explained below.
The benchmark set (c) is called \emph{mining}, and its compositional structure is shown in~\cref{fig:exbenchmark}, where the small circle has weight $0$ and role $\eve$. 
Intuitively, 
the roles $\eve$ and $\adam$, respectively, correspond to an \emph{explorer} and the \emph{environment} that prevents the explorer from earning rewards.
Note that the winner of the entry position may not be determined solely within the uppermost $\mpg{A}$, i.e., the explorer  may choose to go deeper into the cave to maximize the rewards. 
For simplicity, we assume that each floor of the cave is the same oMPG $\mpg{A}$,
which is randomly generated in a non-compositional manner. 
The benchmark set (d) is built in a more random manner as follows.
Each MPG is built inductively from the bottom layer, and each layer randomly chooses a compositional structure from five pre-fixed compositional structures. This continues for $\sim 20$ layers.

The average number of positions is approximately $1.6 \times 10^7$ for the benchmark set (c), and  $1.7 \times 10^7$  for  (d).
We limit the size of games to this order because we could not generate input files for QDPM for larger games. We note that CompMPG could solve larger games  (with $10^8$ positions) constructed in the (c)--(d) ways within at most 5 seconds.


\myparagraph{Results and Discussion}
\begin{figure}[tb]
  \centering
  \begin{subfigure}[]{0.45\columnwidth}
        \centering
        \includegraphics[width=1\textwidth]{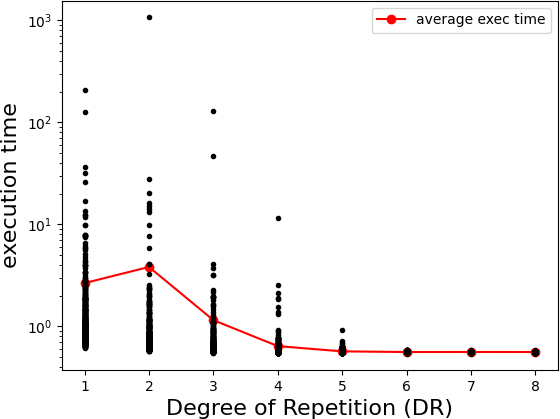}
        \caption{Influence of degree of repetition (DR). Execution time of CompMPG in seconds.}
        \label{subfig:infSeqComp}
  \end{subfigure}
  \hfill
  \begin{subfigure}[]{0.45\columnwidth}
        \centering
        \includegraphics[width=1\textwidth]{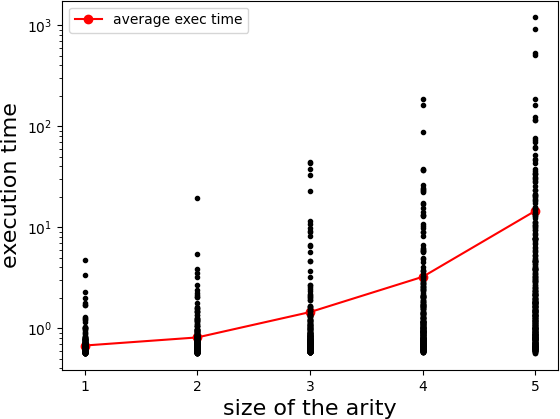}
        \caption{Influence of the size of the arity. Execution time of CompMPG in seconds.}
        \label{subfig:infArity}
  \end{subfigure} 
  \begin{subfigure}[]{0.45\columnwidth}
        \centering
        \includegraphics[width=1\textwidth]{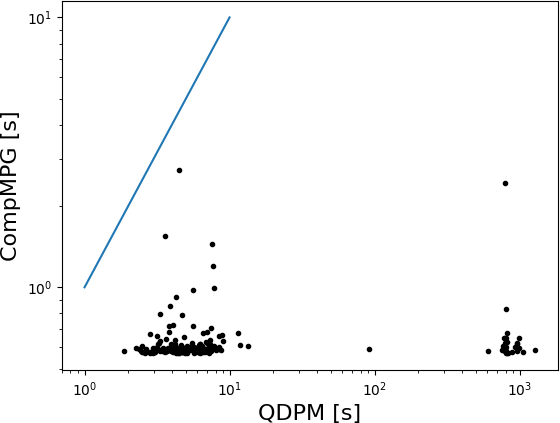}
        \caption{Mining.
        The average number of positions is $16404975$, and the average number of edges is $25847402$.
        The line is $y=x$.
        }
        \label{subfig:mining}
  \end{subfigure}  
  \hfill
  \begin{subfigure}[]{0.45\columnwidth}
        \centering
        \includegraphics[width=1\textwidth]{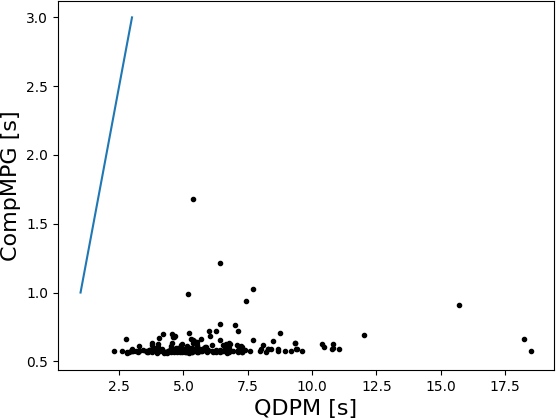}
        \caption{Randomized.
        The average number of positions is $17683982$, and the average number of edges is $29718173$.
        The line is $y=x$.
        }
        \label{subfig:randomized}
  \end{subfigure} 
  \caption{Experimental results for the benchmark sets (a)--(d). }
  \label{fig:results} 
\end{figure}
The experimental results are summarized in~\cref{fig:results}.
The four figures contain all the execution results, with no input MPGs leading to timeout (2000 seconds).
Based on these, we address the research questions as follows.

\begin{description}
    \item[RQ1] The results shown in~\cref{subfig:infSeqComp} and~\cref{subfig:infArity} indicate that the degree of repetition and  and the size of the arity  influences the performance of $\compMPG$. In~\cref{subfig:infSeqComp},
 as DR becomes bigger, CompMPG can exploit more repetition and gets faster. This is as expected.

    As the arity of $\mpg{A}$ increases, the number of dominant (i.e.\ optimal in the multi-objective sense) strategies on $\mpg{A}$ tends to increase. That is: the more exits, the more objectives. This makes the meager semantics less meager, leading to more computational cost. 
    \item[RQ2] All MPGs, especially in (d), are randomly generated with fixed compositional structures, and $\compMPG$ shows stable performance over all benchmarks.
    This means that $\compMPG$ solves a wide variety of MPGs by exploiting compositional structures.
    \item[RQ3] \cref{subfig:mining} shows that $\compMPG$ is remarkably faster than QDPM.
  $\compMPG$ is faster than QDPM for every input: for some MPGs, QDPM takes $600$ to $10^3$ seconds, while $\compMPG$ finishes within 1 second.
    \cref{subfig:randomized}  shows that $\compMPG$ is considerably faster than QDPM for more random (yet structured) benchmarks, too.

    In conclusion, our  CompMDP performed consistently well for large games (with approximately $10^7$ positions), while non-compositional algorithms such as QDPM can struggle. Overall, we clearly see the advantage of compositionality.
\end{description}

\bibliography{arxiv.bib}

\appendix

\section{Definitions on $\roMPG$ and $\oMPG$}
\label{sec:DefOfROMPGandOMPG}

\subsection{Definitions on $\roMPG$}
\label{sec:DefOfROMPG}

\begin{definition}[sum $\oplus$ of roMPGs]
\label{def:sumOfRoMPGs}
Let $\mpg{A}:m\rightarrow n$ and $\mpg{B}:k\rightarrow l$ be roMPGs. Their sum $\mpg{A}\oplus \mpg{B}:m+k\rightarrow n + l$ is given by $\big(m+k, n+l, Q^{\mpg{A}}+Q^{\mpg{B}}, E^{\mpg{A}+\mpg{B}}, [\rho^{\mpg{A}},\rho^{\mpg{B}}], [w^{\mpg{A}}, w^{\mpg{B}}]\big)$, where $E^{\mpg{A}+ \mpg{B}}$ naturally combines the two set of edges by case distinction: 
  \begin{alignat*}{3}
      & \text{for $(s, s')\in (\nset{m} + Q^{\mpg{A}})\times (\nset{n} + Q^{\mpg{A}})$, }&
       \text{$(s, s')\in E^{\mpg{A}+ \mpg{B}}$} & \text{ if  $(s, s')\in E^{\mpg{A}}$,}\\
      & \text{for $(s, s')\in ([m+1, m+k] + Q^{\mpg{B}})\times ([n+1, n+l] + Q^{\mpg{B}})$, }&\text{$(s, s')\in E^{\mpg{A}+ \mpg{B}}$} & \text{ if $(\lplus{m}{s}, \lplus{n}{s'})\in E^{\mpg{B}}$},\\
      & \text{for $s, s'\in \nset{m+k} + \nset{n+l}+Q^{\mpg{A}}+Q^{\mpg{B}}$, } &\text{$(s, s')\not \in E^{\mpg{A}+ \mpg{B}}$} & \text{ otherwise,}
  \end{alignat*}
 where $\lplus{m}{s} \defeq s-m$ if $s\in [m+1, m+k]$, and $\lplus{m}{s} \defeq s$ otherwise. The definition of $[\rho^{\mpg{A}}, \rho^{\mpg{B}}], [w^{\mpg{A}}, w^{\mpg{B}}]$ are similar.

\end{definition}

\begin{definition}[isomorphism of roMPGs]
\label{def:iso_rompgs}
Let $\mpg{A} = (m, n, Q^{\mpg{A}}, E^{\mpg{A}}, \rho^{\mpg{A}}, w^{\mpg{A}})$ and $\mpg{B} = (m, n, Q^{\mpg{B}}, E^{\mpg{B}}, \rho^{\mpg{B}}, w^{\mpg{B}})$ be roMPGs, assuming that they have the same arity $m, n$. An \emph{isomorphism} from $\mpg{A}$ to $\mpg{B}$ is a bijection $\eta:Q^{\mpg{A}}\rightarrow Q^{\mpg{B}}$ that preserves the MPG structure, that is, (i) for each $(s, s')\in (\nset{m}+Q^{\mpg{A}})\times (\nset{n}+Q^{\mpg{A}})$, $(s, s')\in E^{\mpg{A}}\Leftrightarrow (\overline{\eta}(s), \overline{\eta}(s'))\in E^{\mpg{B}}$, (ii) for $s\in Q^{\mpg{A}}$, $\rho^{\mpg{A}}(s) = \rho^{\mpg{B}}(\eta(s))$, and (iii) for $s\in Q^{\mpg{A}}$, $w^{\mpg{A}}(s) = w^{\mpg{B}}(s)$. Here, we extend $\eta$ to $\overline{\eta}:\Nat + Q^{\mpg{A}}\rightarrow \Nat + Q^{\mpg{B}}$ by $\overline{\eta}(n) = n$ for $n\in \Nat$.
\end{definition}


\subsection{Definitions on $\oMPG$}
\label{sec:DefOfOMPG}

In the definition of the (bidirectional) sequential composition $\seqcomp$ of $\oMPG$, all (unidirectional) operations $\seqcomp$, $\oplus$, $\trsyb$ of $\roMPG$ are used. 

\begin{definition}[$\seqcomp$ of oMPGs]
\label{def:seqCompOMPGs}
Let $\mpg{A}:(\nfr{m},\nfl{m})\rightarrow (\nfr{l},\nfl{l})$ and $\mpg{B}:(\nfr{l},\nfl{l})\rightarrow (\nfr{n},\nfl{n})$ be arrows in $\oMPG$. Their \emph{sequential composition} $\mpg{A}\seqcomp\mpg{B}:(\nfr{m},\nfl{m})\rightarrow (\nfr{n},\nfl{n})$ is
defined by 
\begin{math}
   \mpg{A}\seqcomp\mpg{B}
 \defeq \trace{\nfl{l}}{\nfr{m} + \nfl{n}}{\nfr{n}+ \nfl{m}}{\big((\Swap{\nfl{l}}{\nfr{m}}\oplus\Id{\nfl{n}}) \seqcomp (\mpg{A}\oplus \Id{\nfl{n}}) \seqcomp (\Id{\nfr{l}}\oplus \Swap{\nfl{m}}{\nfl{n}}) \seqcomp (\mpg{B}\oplus\Id{\nfl{m}}) \seqcomp (\Swap{\nfr{n}}{\nfl{l}}\oplus\Id{\nfl{m}})\big)}
\end{math}.
\end{definition}

The (bidirectional) sum $\oplus$ of $\oMPG$ is defined using (unidirectional) $\seqcomp$ and $\oplus$ of $\roMPG$.

\begin{definition}[$\oplus$ of oMPGs]
\label{def:sumOMPGs}
Let $\mpg{A}:(\nfr{m},\nfl{m})\rightarrow (\nfr{n},\nfl{n})$ and $\mpg{B}:(\nfr{k},\nfl{k})\rightarrow (\nfr{l},\nfl{l})$ be arrows in $\oMPG$. Their \emph{sum} $\mpg{A}\oplus\mpg{B}:(\nfr{m}+\nfr{k},\nfl{k}+\nfl{m})\rightarrow (\nfr{n}+\nfr{l},\nfl{l}+\nfl{n})$ is
defined by 
\begin{math}
   \mpg{A}\oplus \mpg{B}
 \defeq (\Swap{\nfr{m}}{\nfr{k}}\oplus \Swap{\nfl{l}}{\nfl{n}})\seqcomp (\Id{\nfr{k}}\oplus \mpg{A}\oplus \Id{\nfl{l}})
 \seqcomp (\Swap{\nfr{k}}{\nfr{n}}\oplus\Swap{\nfl{l}}{\nfl{m}}) \seqcomp (\Id{\nfr{n}}\oplus \mpg{B}\oplus \Id{\nfl{m}})
\end{math}.
\end{definition}

\section{Play Monads}
\label{sec:play_monads}
\begin{definition}[play monad]
The \emph{play monad} $(\fpmonad, \eta, \mu)$ on $\Sets$ is defined by $\fpmonad(X) \defeq X + \Real\times X + \{\const{\eve}\} + \{\const{\adam}\}$, $\eta_X(x) \defeq x$, and
\begin{align*}
    \mu_X(z) &\defeq \begin{cases}
        z &\text{ if }z\in X \text{ or }z\in \Real \times X,\\
        \const{\eve} &\text{ if }z = \const{\eve} \text{ or }z\in \Real \times \{\const{\eve}\},\\
        \const{\adam} &\text{ if }z = \const{\adam} \text{ or }z\in \Real \times \{\const{\adam}\},\\
        (r+q, x) &\text{ if } z = (r, q, x) \in \Real \times \Real \times X.
    \end{cases}
\end{align*}
   
\end{definition}

The category $\fpsemCat$ is closely related with the \emph{Kleisli category} $\kleisli{\Sets}{\fpmonad}$, whose arrow $f:X\rightarrow Y$ is a $\fpmonad$-effectful function from $X$ to $\fpmonad(Y)$. We state the relationship between $\kleisli{\Sets}{\fpmonad}$ and $\fpsemCat$ below.

\begin{proposition}
\label{prop:fpsemcatSMC}
    \begin{enumerate}
        \item The category $\fpsemCat$ is a subcategory of $\kleisli{\Sets}{\fpmonad}$ by restricting whose objects are intervals $\nset{n}$ for $n\in \Nat$ and whose arrows satisfy the realizability condition.
        \item The category $\fpsemCat$ inherits \emph{symmetric monoidal structure} from $\kleisli{\Sets}{\fpmonad}$ whose monoidal product is coproduct.
    \end{enumerate}
\end{proposition}


    

\section{Proof of~\cref{prop:fplayTSMC} }
\label{sec:proofFplayTSMC}
Before proving the statement, we explicitly define trace operator in $\fpsemCat$. We introduce \emph{semantic TDP} and its denotation by mimicking~\cref{def:traced_induced_plays} and~\cref{def:denTracedInducedPlays}.

\begin{definition}[semantic TDP]
Let $f:l+m\rightarrow l+n$ in $\fpsemCat$, and $j\in \nset{m}$.
The \emph{semantic TDP} of $f$ from $j$ is the (unique) possibly infinite sequence $v = (v_0,v_1,\cdots)$ of elements in $\nset{l}+\nset{m}+\nset{n} + \Real\times (\nset{l} + \nset{n}) + \{\const{\eve}, \const{\adam}\}$ that $v_0 = i$ and satisfies the following conditions:
\begin{enumerate}
    \item
     $v_{j+1}$ is undefined if $v_j \in \nset{n} + \Real \times \nset{n} + \{\const{\eve}, \const{\adam}\}$,
    \item 
    $v_{j+1} = f(v_i)$ if $v_j\in \nset{l}+\nset{m}$,
    \item
    $v_{j+1} = f(k)$ if $v_j = (r, k)\in \Real \times \nset{l}$.
\end{enumerate}
\end{definition}



\begin{definition}[denotation of semantic TDP]
    Let $f:l+m\rightarrow l+n$, $j\in \nset{m}$, and $v \defeq v_0,v_1,\cdots $ be the semantic TDP of $f$ from $j$. The \emph{denotation} $\denosdigest{v}\in \fpmonad(\nset{n})$ of $v$ is defined as follows:
    \begin{alignat*}{2}
     & \const{\eve} && \text{ if }v \text{ is infinite and satisfies the MP condition,}\\
     & \const{\adam} && \text{ if }v \text{ is infinite and does not satisfy the MP condition,}\\
     & v_k && \text{ if }v = v_0,\cdots, v_k \text{ and }v_k \in \{\const{\eve}, \const{\adam}\},\\
     & v_k       && \text{ if } v = v_0,\cdots, v_k, v_j\in \nset{l} \text{ for each }j\in \nset{k-1}\text{, and } v_k\in \nset{n},\\
     & (\sum^{}_{j\in \nset{k}} \mathrm{wt}(v_j), v_k)&&\text{ if } v = v_0,\cdots, v_k, \text{ there is }v_j\in \Real\times \nset{l}, \text{ and }v_k\in \nset{n},\\
     & (\sum^{}_{j\in \nset{k}} \mathrm{wt}(v_j), \pi_2(v_k))&&\text{ otherwise.} 
    \end{alignat*}
\end{definition}

Finally, we define the \emph{trace} $\trsyb$ in $\fpsemCat$ by using semantic digests.

\begin{definition}[trace operator $\trsyb$ over $\fpsemCat$]
Let $f:l+m\rightarrow l+n$ in $\fpsemCat$. The \emph{trace} $\trace{l}{m}{n}{f}$ of $f$ is given by $\trace{l}{m}{n}{f}(i) \defeq \denosdigest{v^{i}}$ for each $i\in \nset{m}$, where  $v^{i}$ is the semantic TDP of $f$ from $i$.

\end{definition}

\begin{proof}[proof of~\cref{prop:fplayTSMC}] 
We already know that $\fpsemCat$ is a symmetric monoidal category by~\cref{prop:fpsemcatSMC}. 
We directly prove that the operator $\trsyb$ satisfies axioms of trace operators. We do not show every axioms, since we can prove other axioms in the same manner.

First, we prove the naturality in $m$. Let $f:l+m'\rightarrow l+n$, and $g: m\rightarrow m'$. We prove that $\trace{l}{m}{n}{(id_l\oplus g)\seqcomp f} = g\seqcomp \trace{l}{m'}{n}{f}$. Let $i\in \nset{m}$, and $v^{i} = (v_k)_{k\in I}$ be the semantic TDP of $(id_l\oplus g)\seqcomp f$ from $i$.
\begin{description}
\item[$\lbrack\text{Case }g(i) = \const{\eve} \text{ or }g(i) = \const{\adam}\rbrack $] By the definition of the sequential composition, $\big(g\ \seqcomp\ \trace{l}{m'}{n}{f}\big)(i)\allowbreak = g(i)$. It is also trivial that $\denosdigest{v^{i}} = g(i)$ by definition. Therefore, $\trace{l}{m}{n}{(id_l\oplus g)\seqcomp f}(i) = \big(g\seqcomp \trace{l}{m'}{n}{f}\big)(i)$.
\item[$\lbrack\text{Case }g(i) = j\in \nset{n}\rbrack $] By definition, $\big(g\seqcomp \trace{l}{m'}{n}{f}\big)(i) = \trace{l}{m'}{n}{f}(j)$. Let $u^{j} = (u_k)_{k\in J}$ be the semantic TDP of $f$ from $j$. Then, for each $k\in J$, $u_k = v_{k+1}$. Since $v_0 = i$,  $\denosdigest{v^{i}} = \denosdigest{u^{j}}$ holds, which means that $\trace{l}{m}{n}{(id_l\oplus g)\seqcomp f}(i) =  g\seqcomp \trace{l}{m'}{n}{f}(i) $. 
\item[$\lbrack\text{Case }g(i) = (r, j)\in \Real\times \nset{n}\rbrack $] Same.

Next, we prove the dinaturality in $l$. Let $f: l+m\rightarrow l'+n$, and $g:l'\rightarrow l$. We prove that $\trace{l}{m}{n}{f\seqcomp (g\oplus \id_n)} = \trace{l'}{m}{n}{(g\oplus \id_m )\seqcomp f}$. Let $i\in \nset{m}$, $v^{i} = (v_k)_{k\in I}$ be the semantic TDP of $f\seqcomp (g\oplus \id_n)$, and $u^{i} = (u_k)_{k\in J}$ be the semantic TDP of $(g\oplus \id_m )\seqcomp f$. 
\item[$\lbrack\text{Case } I \text{ is finite}\rbrack $] Straightforward. 
\item[$\lbrack\text{Case } I \text{ is infinite}\rbrack $] It is easy to prove that $J$ is also infinite, and $v^{i}$ and $u^{i}$ are periodic, i.e., there are $k_1, k_2, l_1, l_2\in \Nat$ such that for all $j > l_1$, $v_{j} = v_{\big((j-l_1)\% k_1\big) + l_1  }$, and  for all $j > l_2$, $u_{j} = u_{\big((j-l_2)\% k_2\big) + l_2  }$. Then, it is also easy to prove that $\sum_{l_1 < j \leq l_1 + k_1} \mathrm{wt}(v_{j}) \geq 0 $ iff $\sum_{l_2 < j \leq l_2 + k_2} \mathrm{wt}(u_{j}) \geq 0 $. By the same argument as~\cref{lem:periodicity}, we conclude that $\denosdigest{v^{i}} = \denosdigest{u^{i}}$, which means that $\trace{l}{m}{n}{f\seqcomp (g\oplus \id_n)}(i) = \trace{l'}{m}{n}{(g\oplus \id_m )\seqcomp f}(i)$.
 

\end{description}
\end{proof}

\section{Meager Semantics}
\label{sec:meagerSemantics}
In this section, we define the \emph{meager semantics} used in $\compMPG$.
\begin{definition}[meager semantic category for plays]
The category $\mpsemCat$  is defined as follows. Its object is a natural number. Its arrow from $m$ to $n$ ($m\rightarrow n$) is a function $f$ from $\nset{m}$ to $\fpmonad(\nset{n}$ that satisfies the \emph{realizability condition} defined as follows:
\begin{itemize}
    \item (Realizability) for each $i, j\in \nset{m}$, if $j\not = i$, then $f(i)\not \in \nset{n}$ or $f(j)\not\in \{f(i)\} + \Real \times \{f(i)\}$.
\end{itemize}
In addition, for each $m, n\in \Nat$, the order $\leq_{m, n}$ in $\mpsemCat(m, n)$ is given by the functional order induced by $\leq_{\fpmonad(\nset{n})}$ in $\fpmonad(\nset{n})$, where $\leq_{\fpmonad(\nset{n})}$ is the least order satisfying the following condition: 
\begin{align*}
    (r_1, i) &\leq_{\fpmonad(\nset{n})} (r_2, i) &&\text{ if }r_1 \geq r_2, \text{ and }i\in \nset{n},\\
    \const{\eve} &\leq_{\fpmonad(\nset{n})} z  &&\text{ if }z\in\fpmonad(\nset{n}),\\
    z&\leq_{\fpmonad(\nset{n})} \const{\adam} && \text{ if }z\in\fpmonad(\nset{n}).\\   
\end{align*}
\end{definition}

The meager semantics of plays can be obtained in the same way. 

\begin{proposition}[a TSMC $\mpsemCat$]
The category $\mpsemCat$ is a TSMC.
\qed
\end{proposition}

\begin{proposition}[$\seqcomp, \oplus, \trsyb$ of  $\mpsemCat$  are monotone]
The category $\mpsemCat$ is $\Ord$-enriched, i.e., the sequential composition $\seqcomp$ is monotone. Moreover, the sum $\oplus$ and the trace operator $\trsyb$ are monotone.
\qed
\end{proposition}

\begin{definition}[$\mpwpFunctor$]
    The \emph{solution functor} $\mpwpFunctor:\roPlay\rightarrow \mpsemCat$ is defined as follows: the mapping of objects is given by $\mpwpFunctor(m) \defeq m$, and for an arrow $\mpg{C}\in \roPlay(m, n)$, we define $\mpwpFunctor(\mpg{C})$ by $\mpwpFunctor(\mpg{C})(i)\defeq \denoplay{\indplaya{\mpg{C}}{i}}{\mpg{C}}$, using the denotation of the (unique) play $\pi^{\mpg{C}}_i$ from entrance $i$.
\end{definition}

\begin{theorem}[$\mpwpFunctor$ is compositional]
\label{thm:mpwpfuncComp}
The data $\mpwpFunctor$ is compositional, that is, $\mpwpFunctor(\mpg{C}\seqcomp \mpg{D}) = \mpwpFunctor(\mpg{C})\seqcomp\mpwpFunctor(\mpg{D}) $, $\mpwpFunctor(\mpg{C}\oplus \mpg{D}) = \mpwpFunctor(\mpg{C})\oplus\mpwpFunctor(\mpg{D}) $, and $\trace{l}{m}{n}{\mpwpFunctor(\mpg{E})} = \mpwpFunctor(\trace{l}{m}{n}{\mpg{E}})$.
\end{theorem}

Next, we define the meager semantics of roMPGs. 
Instead of the powerset functor, we introduce the \emph{maximal functor} for representing optimal $\adam$-strategies.

\begin{definition}[incomparable sets]
Let $X$ be a ordered set. The ordered set $\big(\uncompfunc(X), \leq_{\uncompfunc(X)}\big)$  of incomparable sets is given by 
\begin{align*}
    \uncompfunc(X) &\defeq \{ S\subseteq X \mid \forall x_1, x_2\in S,\ x_1\leq_X x_2 \Rightarrow x_1 = x_2,\ S\not = \emptyset,\text{ and }S\text{ is finite}\},\\ 
    \leq_{\uncompfunc(X)} &\defeq \{ (S, T) \mid \forall x_1\in S,\ \exists x_2\in T,\ x_1\leq_X x_2\}
\end{align*}
\end{definition}

\begin{proposition}
    Let $X$ be a ordered set. The binary relation $\leq_{\uncompfunc(X)}$ is an order on $\uncompfunc(X) $.
\end{proposition}
\begin{proof}
We directly prove each axioms:
\begin{description}
    \item[$\lbrack$reflexivity$\rbrack$] Trivial.
    \item[$\lbrack$antisymmetry$\rbrack$] Suppose that $S \leq_{\uncompfunc(X)} T$ and $T \leq_{\uncompfunc(X)} S$. For each $s\in S$, there is $t\in T$ such that $s\leq_X t$. There is also $s'\in S$ such that $t\leq_X s'$. Then, $s\leq_X s'$ holds, which means that $s=s'$, since the elements in $S$ are incomparable. Thus, we conclude that $s = t$, and $S\subseteq T$. By the same argument, we conclude that $S = T$.
    \item[$\lbrack$transitivity$\rbrack$] Let $S \leq_{\uncompfunc(X)} T$ and $T \leq_{\uncompfunc(X)} U$. For each $s\in S$, there is $t\in T$ such that $s\leq_X t$. There is also $u\in U$ such that $t\leq_X u$, therefore we conclude that $s\leq_X u$ and $S\leq_{\uncompfunc(X)} U$.
\end{description}
\end{proof}

\begin{definition}
  Let $X = (|X|, \leq_X)$ be an ordered set, and $S\in \mathcal{P}(|X|)$.
  We define a set $S^{\circ}$ as $\{x \in S \mid x\text{ is maximal in }S\}$.
\end{definition}


\begin{definition}[maximal functor]
The \emph{maximal} functors $\maxfunc: \Ord\rightarrow \Ord$ are given by $\maxfunc(X) \defeq (\uncompfunc(X), \leq_{\uncompfunc(X)})$, and $ \maxfunc(f)(S) \defeq \big(f(S)\big)^{\circ}$.
\end{definition}

\begin{proposition}
    The data $\maxfunc$ is a functor from $\Ord$ to $\Ord$.
\end{proposition}
\begin{proof}
First, we prove that $\maxfunc$ is well-defined, i.e., $\maxfunc(f)$ is monotone. Let $f:X\rightarrow Y$, $S_1, S_2\in \uncompfunc(X)$ and $S_1\leq S_2$. For each $y_1\in \maxfunc(f)(S_1)$,  there is $x_1\in S_1$ such that $f(x_1) = y_1$, and there is $x_2\in S_2$ such that $x_1\leq x_2$. If $f(x_2) \not \in \maxfunc(f)(S_2)$, then there is $x'_2\in S_2$ such that $f(x_2) < f(x'_2)$ and $f(x'_2) \in \maxfunc(f)(S_2)$ because $S_2$ is a finite set. Then, $y_1 = f(x_1) < f(x'_2)$ holds. Therefore, we conclude that $\maxfunc(f)(S_1)\leq \maxfunc(f)(S_2)$.

Since $\maxfunc(\id_X) = \id_X$ is trivial by definition, we prove that $\maxfunc$ preserves sequential composition of arrows and finish the proof.

 Let $f: X\rightarrow Y$, $g: Y\rightarrow Z$. We prove that $\maxfunc(g) \circ \maxfunc(f) = \maxfunc(g\circ f)$. First, we prove that $\maxfunc(g) \circ \maxfunc(f) \subseteq \maxfunc(g\circ f)$, and then we prove that $\maxfunc(g\circ f)\subseteq \maxfunc(g) \circ \maxfunc(f)$.

Let $S\in \uncompfunc(X)$, and $z\in \maxfunc(g)\circ \maxfunc(f)(S)$. There is $x\in S$ such that $g\circ f(x) = z$,  $f(x)$ is maximal in $\maxfunc(f)(S)$, and $g(f(x))$ is maximal in $\maxfunc(g)\big(\maxfunc(f)(S)\big)$. Assume that there is $x'\in S$ such that $g\circ f(x') > g\circ f(x) $. We can assume that $f(x')\in \maxfunc(f)(S)$ because if $f(x')\not \in \maxfunc(f)(S)$, there is a $x''\in S$ such that $f(x'')\in  \maxfunc(f)(S)$ and $f(x'') > f(x')$, which means that $g\circ f(x'')\geq g\circ f(x')$. Since $g(f(x')) > g(f(x))$, $g(f(x))$ is not maximal in $\maxfunc(g)\big(\maxfunc(f)(S)\big)$, which leads to the contradiction. Therefore, we can conclude that $g\circ f(x)$ is maximal in $\maxfunc(g\circ f)(S)$, and $\maxfunc(g)\circ \maxfunc(f)(S)\subseteq \maxfunc(g\circ f)(S)$.

    Suppose that $z\in \maxfunc(g\circ f)(S)$. There is $x\in S$ such that $z = g\circ f(x)$ and $z$ is maximal in $ g\circ f(S)$. Suppose that $f(x)\in \maxfunc(f)(S)$ and $g(f(x)) \not \in \maxfunc(g)\big(\maxfunc(f)(S)\big)$. Then, there is $x'\in S$ and   $f(x')\in \maxfunc(f)(S)$ and $g(f(x')) \in \maxfunc(g)\big(\maxfunc(f)(S)\big)$ and $g(f(x')) > g(f(x))$ because $S$ is a finite set. But this contradicts to the fact that $g(f(x))$ is maximal in $\maxfunc(g\circ f)(S)$. Thus, $g(f(x))$ is also maximal in $\maxfunc(g)\big(\maxfunc(f)(S)\big)$. Suppose that $f(x) \not \in  \maxfunc(f)(S)$. Then, there is $x'\in S$ such that $f(x') \in  \maxfunc(f)(S)$ and $f(x') > f(x)$ because $S$ is a finite set. Then, $g\circ f(x') = g\circ f(x)$ because $g$ is monotone and $g\circ f(x)$ is maximal in $ g\circ f(S)$.  By the same argument, we can conclude that $g(f(x')) \in \maxfunc(g)\big(\maxfunc(f)(S)\big)$. Thus, $\maxfunc(g\circ f)(S)\subseteq \maxfunc(g)\circ \maxfunc(f)(S)$.

\end{proof}

\begin{proposition}
    The functor $\maxfunc$ is a lax monoidal functor from $(\Ord, \times , (\mathbf{1}, =))$ to $(\Ord, \times , (\mathbf{1}, =))$.
\end{proposition}

We use the following lemmas for proving the axioms of the trace operator.
\begin{lemma}
\label{lem:seqcomp_monotone_and_maximal}
    Let $S\subseteq \mpsemCat(m, l)$ and $T\subseteq \mpsemCat(l, n)$, and $S, T$ be finite sets. Then, $\{f\seqcomp g\mid f\in S, \text{ and }g\in T \}^{\circ} = \{f\seqcomp g \mid f\in S^{\circ}\text{ and }g\in T^{\circ} \}^{\circ}$.
\end{lemma}
\begin{proof}
    Let $f_1\seqcomp g_1\in \{f\seqcomp g\mid f\in S, \text{ and }g\in T \}^{\circ}$. If $f_1\not \in S^{\circ}$, then $f_2 \in S^{\circ}$ such that $f_1 < f_2$. Similarly, if $g_1\not \in T^{\circ}$, then $g_2 \in T^{\circ}$ such that $g_1 < g_2$. Then, $f_1\seqcomp g_1 \leq f_2\seqcomp g_2$. Since  $f_1\seqcomp g_1$ is maximal, $f_1\seqcomp g_1 = f_2\seqcomp g_2$, thus $\{f\seqcomp g\mid f\in S, \text{ and }g\in T \}^{\circ} \subseteq \{f\seqcomp g \mid f\in S^{\circ}\text{ and }g\in T^{\circ} \}^{\circ}$. 

     Let $f_1\seqcomp g_1\in \{f\seqcomp g \mid f\in S^{\circ}\text{ and }g\in T^{\circ} \}^{\circ}$. For each $f_2\in S$ and $g_2\in T$, there are $f_3\in S^{\circ}$ and $g_3\in T^{\circ}$ such that $f_2\leq f_3$ and $g_2\leq g_3$. Since the sequential composition $\seqcomp$ is monotone, $f_2\seqcomp g_2\leq f_3\seqcomp g_3$. By assumption, $f_3\seqcomp g_3 \not > f_1\seqcomp g_1$, which means that $f_2\seqcomp g_2 \not > f_1\seqcomp g_1$ and $f_1\seqcomp g_1 \in \{f\seqcomp g\mid f\in S, \text{ and }g\in T \}^{\circ}$. Therefore, $\{f\seqcomp g \mid f\in S^{\circ}\text{ and }g\in T^{\circ} \}^{\circ}\subseteq \{f\seqcomp g\mid f\in S, \text{ and }g\in T \}^{\circ}$.
\end{proof}

\begin{lemma}
\label{lem:sum_monotone_and_maximal}
    Let $S\subseteq \mpsemCat(m, n)$ and $T\subseteq \mpsemCat(k, l)$, and $S, T$ be finite sets. Then, $\{f\oplus g\mid f\in S, \text{ and }g\in T \}^{\circ} = \{f\oplus g \mid f\in S^{\circ}\text{ and }g\in T^{\circ} \}^{\circ}$.
\end{lemma}

\begin{lemma}
\label{lem:trace_monotone_and_maximal}
    Let $S\subseteq \mpsemCat(l+m, l+n)$, and $S$ be finite sets. Then, $\{\trace{l}{m}{n}{f}\mid f\in S\}^{\circ} = \{\trace{l}{m}{n}{f} \mid f\in S^{\circ} \}^{\circ}$.
\end{lemma}

By~\cref{lem:seqcomp_monotone_and_maximal,lem:sum_monotone_and_maximal,lem:trace_monotone_and_maximal}, we can easily prove the following proposition. Here, $F_{\star}$ denote change of base with lax monoidal functor $F$.

\begin{proposition}
    The category $\maxfunc_{\star}(\mpsemCat)$ is TSMC.
\end{proposition}

\begin{proposition}[$\oplus, \trsyb$ of  $\maxfunc_{\star}(\mpsemCat)$  are monotone]
The sum $\oplus$ and the trace operator $\trsyb$ of $\maxfunc_{\star}(\mpsemCat)$ are monotone.
\end{proposition}
\begin{proof}
The sum $\oplus$ is obviously monotone. We prove that $\trsyb$ is monotone. Let $S, T\in \maxfunc_{\star}(\mpsemCat)(l+m, l+n)$, and $S\leq T$. For $f\in S$ such that $\trace{l}{m}{n}{f}\in \trace{l}{m}{n}{S}$, there is $g\in T$ such that $f\leq g$. If $\trace{l}{m}{n}{g}\not\in \trace{l}{m}{n}{T}$, there is $g'\in T$ such that $\trace{l}{m}{n}{g'}\in \trace{l}{m}{n}{T}$ and $\trace{l}{m}{n}{g} < \trace{l}{m}{n}{g'}$. Then, $\trace{l}{m}{n}{f} \leq \trace{l}{m}{n}{g} < \trace{l}{m}{n}{g'}$ holds. Therefore, $ \trace{l}{m}{n}{S}\leq  \trace{l}{m}{n}{T}$ holds.
\end{proof}


Next, we introduce the \emph{minimal functor} for representing optimal $\eve$-strategies.

\begin{definition}
  Let $X = (|X|, \leq_X)$ be an ordered set, and $S\in \mathcal{P}(|X|)$.
  We define a set $S_{\circ}$ as $\{x \in S \mid x\text{ is minimal in }S\}$.
\end{definition}

\begin{definition}[minimal functor]
The \emph{minimal} functors $\minfunc: \Ord\rightarrow \Sets$ are given by $\minfunc(X) \defeq \uncompfunc(X)$, and $ \minfunc(f)(S) \defeq \big(f(S)\big)_{\circ}$.
\end{definition}

\begin{remark}
     The codomain of $\minfunc$ cannot be $\Ord$ unlike $\maxfunc$, because there is a monotone function $f$ such that $\minfunc(f)$ is not monotone. Let $X \defeq (\{x_1, x_2\}, = )$ and $Y\defeq (\{y_1, y_2\}, \leq_Y)$, where $y_1< y_2$, and $f:X\rightarrow Y$ be $f(x_1) \defeq y_1$ and $f(x_2) \defeq y_2$. Suppose that $S \defeq  \{x_2\} $ and $T\defeq \{x_1, x_2\}$. Obviously $S\leq T$, but $\minfunc(f)(S)\not \leq \minfunc(f)(T)$ since $\minfunc(f)(S) = \{y_2\}$ and $\minfunc(f)(T) = \{y_1\}$.
     Note that  $\maxfunc(f)(S)\leq \maxfunc(f)(T)$ holds since $\maxfunc(f)(S) = \{y_2\}$ and $\maxfunc(f)(T) = \{y_2\}$.
\end{remark}

\begin{proposition}
    The functor $\minfunc$ is a lax monoidal functor from $(\Ord, \times , (\mathbf{1}, =))$ to $(\Sets, \times , \mathbf{1})$.
\end{proposition}


\begin{definition}
    The category $\mrsemCat$ is $\minfunc_{\star}\big(\maxfunc_{\star}(\mpsemCat)\big)$. Concretely, let $F:m\rightarrow l$, $G:l\rightarrow n$ be arrows in $\mrsemCat$. Their sequential composition $F\seqcomp G$ is given by $F\seqcomp G \defeq \big\{ \{f\seqcomp g\mid f\in F',\ g\in G'\}^{\circ}\mid F'\in F,\ G'\in G\big\}_{\circ}$, where $f\seqcomp g$ is the sequential composition in $\mpsemCat$.
\end{definition}


\begin{lemma}
    Let $S\subseteq \maxfunc_{\star}(\mpsemCat)(m, l)$ and $T\subseteq \maxfunc_{\star}(\mpsemCat)(l, n)$. Also let $S, T$ be finite sets. Then, $\{f\seqcomp g\mid f\in S, \text{ and }g\in T \}_{\circ} = \{f\seqcomp g \mid f\in S_{\circ}\text{ and }g\in T_{\circ} \}_{\circ}$.
\end{lemma}
\begin{proof}
    Let $f_1\seqcomp g_1\in \{f\seqcomp g\mid f\in S, \text{ and }g\in T \}_{\circ}$. If $f_1\not \in S_{\circ}$, then $f_2 \in S_{\circ}$ such that $f_1 > f_2$. Similarly, if $g_1\not \in T_{\circ}$, then $g_2 \in T_{\circ}$ such that $g_1 > g_2$. Then, $f_1\seqcomp g_1 \geq f_2\seqcomp g_2$. Since  $f_1\seqcomp g_1$ is minimal, $f_1\seqcomp g_1 = f_2\seqcomp g_2$, thus $\{f\seqcomp g\mid f\in S, \text{ and }g\in T \}_{\circ} \subseteq \{f\seqcomp g \mid f\in S_{\circ}\text{ and }g\in T_{\circ} \}_{\circ}$. 

     Let $f_1\seqcomp g_1\in \{f\seqcomp g \mid f\in S_{\circ}\text{ and }g\in T_{\circ} \}_{\circ}$. Since the sequential composition $\seqcomp$ is monotone, $f_1\seqcomp g_1 \in \{f\seqcomp g\mid f\in S, \text{ and }g\in T \}_{\circ}$. Therefore, $\{f\seqcomp g \mid f\in S_{\circ}\text{ and }g\in T_{\circ} \}_{\circ}\subseteq \{f\seqcomp g\mid f\in S, \text{ and }g\in T \}_{\circ}$.
\end{proof}

\begin{lemma}
\label{lem:sum_monotone_and_minimal}
    Let $S\subseteq  \maxfunc_{\star}(\mpsemCat)(m, n)$ and $T\subseteq  \maxfunc_{\star}(\mpsemCat)(k, l)$, and $S, T$ be finite sets. Then, $\{f\oplus g\mid f\in S, \text{ and }g\in T \}_{\circ} = \{f\oplus g \mid f\in S_{\circ}\text{ and }g\in T_{\circ} \}_{\circ}$.
\end{lemma}

\begin{lemma}
\label{lem:trace_monotone_and_minimal}
    Let $S\subseteq  \maxfunc_{\star}(\mpsemCat)(l+m, l+n)$, and $S$ be finite sets. Then, $\{\trace{l}{m}{n}{f}\mid f\in S\}_{\circ} = \{\trace{l}{m}{n}{f} \mid f\in S_{\circ} \}_{\circ}$.
\end{lemma}

\begin{proposition}
    $\mrsemCat$ is a TSMC.
    \qed
\end{proposition}

Finally, we introduce the \emph{meager rightward winning-position functor}.
\begin{definition}
The \emph{rightward winning-position functor} $\mrwpFunctor$ is defined as follows: the mapping on objects is given by $\mrwpFunctor(m) \defeq m$, and for an arrow $\mpg{A}\in \roMPG(m, n)$, we define $\mrwpFunctor(\mpg{A})$ by 
\begin{align*}
    \mrwpFunctor(\mpg{A})\defeq \Big\{ \big\{  \mpwpFunctor\big(\sybpg{\mpg{A}}{\tau^{\eve}}{\tau^{\adam}}\big)\mid \tau_{\adam}\in  \stra{\adam}{\mpg{A}}\big\}^{\circ}\mid \tau_{\eve}\in \stra{\eve}{\mpg{A}}\Big\}_{\circ}.
\end{align*}
\end{definition}

\end{document}